\documentclass[twoside,preprint,11pt]{article}

\usepackage{blindtext}
\usepackage{jmlr2e}

\usepackage{enumerate}
\usepackage{url} 
\usepackage{float}

\usepackage{multirow}
\usepackage{lipsum}
\usepackage{amsfonts}

\usepackage{amssymb,calc}
\usepackage{mathrsfs}
\usepackage{mathtools}
\usepackage{bm}
\usepackage{bbm}
\usepackage{subfigure}

\usepackage{enumerate}
\usepackage{rotating}
\RequirePackage{cleveref}
\usepackage{hyphenat}
\usepackage{tikz}
\usepackage{pgfplots}

\usepackage[T1]{fontenc}
\usepackage{textcomp} 
\usepackage{lmodern}

\usepackage{colortbl}

\usepackage{booktabs}

\usepackage{algorithm}

\makeatletter
\renewcommand{\ALG@name}{Algorithm}
\makeatother

\usepackage{bm}

\newcommand{\wh}[1]{\widehat{#1}}

\def\C {\,|\:}

\newcommand\mN{\mathcal N}

\newcommand\mG{{\mathcal G}}
\newcommand\mF{{\mathcal F}}

\newcommand\mX{{\mathcal X}}

\newcommand\btheta{{\bm \theta}}

\newcommand\bphi{{\bm \phi}}

\newcommand\E{\mathbb E}
\newcommand\bomega{{\bm\omega}}

\renewcommand\d{{\mathrm d}}

\renewcommand\P{\mathbb P}

\newcommand\mL{\mathcal L}
\newcommand\mA{\mathcal A}
\newcommand\Xn{X^{(n)}}
\newcommand\R{\mathbb R}
\renewcommand\b{{\bm{\beta}}}
\newcommand\e{\mathrm e}

\newcommand\1{{\mathbf 1}}

\newcommand\iid{\overset{\text{iid}}{\sim}}

\newcommand{\KL}{{\text{KL}}}

\newcommand{\norm}[1]{\left\Vert#1\right\Vert}
\newcommand{\abs}[1]{\left\vert#1\right\vert}

\newcommand{\wt}[1]{\widetilde{#1}}

\def\argmin{\mathop{\arg\min}}

\newtheorem{thm}{Theorem}

\def\spacingset#1{\renewcommand{\baselinestretch}%
{#1}\small\normalsize} 

\usepackage{lastpage}
\jmlrheading{27}{2026}{1-\pageref{LastPage}}{7/23; Revised
10/25}{2/26}{23-0946}{Yuexi Wang and Veronika Ro\v{c}kov\'a}
\ShortHeadings{Generative Bayesian Inference with GANs}{Wang and Rockova}
\firstpageno{1}

\begin{document}

\title{Generative Bayesian Inference with GANs}

\author{\name Yuexi Wang \email yxwang99@illinois.edu \\
\addr Department of Statistics \\
University of Illinois Urbana-Champaign\\
Champaign, IL 61820, USA 
\AND Veronika Ro\v{c}kov\'a \email veronika.rockova@chicagobooth.edu\\
       \addr Booth School of Business\\
       University of  Chicago\\
       Chicago, IL 60637, USA}

\editor{Daniel Roy}

\maketitle

\begin{abstract}
In the absence of explicit or tractable likelihoods, Bayesians often resort to approximate Bayesian computation (ABC) for inference.
Our work bridges ABC with deep neural implicit samplers based on generative adversarial networks (GANs) and adversarial variational Bayes.
Both ABC and GANs   compare aspects of observed and fake data to simulate from posteriors  and likelihoods, respectively. We develop a Bayesian GAN (B-GAN) sampler that directly targets the posterior by solving an adversarial optimization problem. B-GAN is driven by a deterministic mapping learned on the ABC reference by {\em conditional} GANs.
Once the mapping has been trained, iid posterior samples are obtained by filtering noise at a negligible additional cost.
We propose two post-processing local refinements using (1) data-driven proposals with importance reweighting, and (2) variational Bayes. 
We support our findings with frequentist-Bayesian results, showing that the typical total variation distance between the true and approximate  posteriors converges to zero for certain neural network generators and discriminators.  Our findings on simulated data show highly competitive performance relative to some of the most recent  likelihood-free posterior simulators.
\end{abstract}

\begin{keywords}
Approximate Bayesian Computation, Generative Adversarial Networks, Implicit Models, Likelihood-free Bayesian Inference, Variational Bayes
\end{keywords}

 \section{ABC and Beyond} 
For a practitioner, much of the value of the Bayesian inferential approach hinges on the ability to compute the entire posterior distribution. 
Very often, it is easier to infer data-generating probability distributions through  simulator models rather than   likelihood functions. However, Bayesian computation  with simulator models can be particularly grueling.

We assume that $\theta\in\Theta\subset\R^d$ is a parameter controlling a simulator-based model that gives rise to a data vector  $\Xn_\theta=(X_1,\dots,X_n)'\sim P_\theta^{(n)}$ which is {\em not} necessarily iid. The model may be provided by a probabilistic program that can be easily simulated but its implicit likelihood $p_\theta^{(n)}=\pi(\Xn\C\theta)$ cannot be evaluated. 
For an unknown inferential target $\theta_0\in\Theta$,  our goal is to approximate the  post-data  inferential density (i.e. the posterior)
 \begin{equation}
 \pi(\theta\C\Xn_{0})\propto p_\theta^{(n)}(\Xn_0)\pi(\theta),\label{eq:post}
 \end{equation}
where $\Xn_0 \sim P_{\theta_0}^{(n)}$ denotes the observed data. We allow for the possibility that {\em both} the likelihood $p_\theta^{(n)}(\cdot)$ and/or the prior $\pi(\theta)$ are analytically intractable but easy to draw from.


Without the obligation to build a model, Approximate Bayesian Computation (ABC) \citep{beaumont2002approximate, sisson2018overview}   provides an approximation to the posterior \eqref{eq:post} by matching aspects of observed and fake data.
This is accomplished via forward simulation of the so-called ABC reference table $\{(\theta_j,\Xn_j)\}_{j=1}^T$ where $\theta_j$'s have been sampled from the prior $\pi(\theta)$ and fake data $\Xn_j$'s have been sampled from the likelihood $p_{\theta_j}^{(n)}(\cdot)$. In order to keep only plausible parameter draws, this table is then  filtered through an accept/reject mechanism to weed out parameter values $\theta_j$ for which the summary statistics of the fake and observed data  were too far.
Our work, albeit not being an ABC method per-se, builds off of recent ABC and simulation-based Bayesian inference innovations described below.


ABC Regression adjustment \citep{beaumont2002approximate, beaumont2003estimation, blum2010non} is a post-processing step  that re-weights and re-adjusts the location of $\theta_j$'s gathered by rejection ABC 
by fitting a (weighted) regression model  of $\theta_j$'s onto summary statistics $s_j=s(\Xn_j)$. 
Such a  model can be regarded as provisional density estimator of $\pi(\theta\C \Xn)$ derived from  $s(\Xn)$ under certain regression distributional assumptions.
More flexible conditional density  estimators, such as neural mixture density  networks  \citep{papamakarios2016fast, lueckmann2017flexible}, have been successfully integrated into ABC without the burden of choosing summary statistics.  Our approach is related to these developments. However,  we do not attempt to learn a flexible parametric approximation  to the   posterior (or the likelihood \citep{lueckmann2019likelihood, papamakarios2019sequential}). 
Instead, we find an {\em implicit} neural sampler from an approximation to $\pi(\theta\C \Xn)$  by training  Generative Adversarial Networks \citep{goodfellow2016deep} on the ABC reference table. 
GANs have been originally conceived to simulate from complex likelihoods by contrasting observed and fake data. ABC, on the other hand,  contrasts observed and fake data  to simulate from complex posteriors.
Bringing together these two approaches,  we propose  the B-GAN posterior sampler, an incarnation of  conditional GANs \citep{gauthier2014conditional, mirza2014conditional,athey2021using,zhou2022deep} for likelihood-free  Bayesian simulation. By contrasting the ABC reference table with a fake dataset under the same marginal distribution $\pi(\Xn)$, B-GAN learns to simulate from an approximation to the conditional distribution $\pi(\theta\C\Xn)$.  Similarly as \citet{papamakarios2016fast} and  \citet{lueckmann2017flexible}, our method is also global in the sense that it learns $\pi(\theta\C \Xn)$ for {\em any} $\Xn$, not necessarily $\Xn_0$. 
More perfected posterior reconstructions can be  obtained with post-processing steps that zoom in onto the posterior distribution evaluated at $\Xn_0$. We consider two such refinements based on: (1) reinforcement learning with  importance sampling,  and (2)  adversarial variational Bayes. We describe each approach below.


Simple rejection ABC   may require exceedingly many trials to obtain only a few accepted samples when the posterior $\pi(\theta\C \Xn_0)$ is much narrower than the prior $\pi(\theta)$
 (see e.g. \citet{marjoram2003markov, sisson2007sequential, beaumont2009adaptive}).  This has motivated query-efficient ABC techniques which intelligently decide where to propose next  (see \citet{jarvenpaa20batch, hennig2012entropy}  for decision-theoretic reasoning or \citet{jarvenpaa2019efficient} and \citet{gutmann2016bayesian} for implementations based on Bayesian optimization and surrogate models). 
Alternatively, \citet{lueckmann2017flexible} learn a Bayesian mixture density network approximating the posterior over multiple rounds of adaptively chosen simulations and use more flexible proposal distributions (not necessarily the prior) with a built-in importance-reweighting scheme. A similar strategy was used in \citet{papamakarios2019sequential} who used  a pilot run of mixture density networks   to learn the proposal distribution for the next round. 
Although $\Xn_0$ {\em is not} used in B-GAN training, it can be used in the proposal inside the ABC reference table. 
Similarly as in \citet{papamakarios2019sequential}, we use $\Xn_0$ to construct a flexible proposal (i.e. an empirical Bayes prior) and convert the draws to posterior samples under the original prior by importance reweighting. 
This  'reinforcement learning' refinement substantially improves the reconstruction accuracy and can be justified by theory.  

Our vanilla {B-GAN} sampler uses contrastive learning \citep{gutmann2018likelihood,durkan2020contrastive} to estimate the conditional distribution $\pi(\theta\C\Xn)$ for {\em any} $\Xn$. 
Since $\Xn_0$ is used only at the evaluation stage (not the training stage), we can custom-make the sampler to $\Xn_0$ by using the B-GAN output {(or the output after reinforcement learning)} as an initialization for implicit variational Bayes   optimization  \citep{tran2017hierarchical, huszar2017variational}.
 Implicit variational Bayes attempts to approximate the posterior using densities which are defined implicitly by a push-forward mapping.  B-GAN also trains such a mapping but the generator will have never seen observed data. At later stages of B-GAN training, we can thereby modify the objective function for the generator so that it minimizes a lower bound to the marginal likelihood. Since the likelihood cannot be evaluated,  we use contrastive learning inside the variational objective to compute the lower bound \citep{huszar2017variational, tran2017hierarchical}. We consider  the joint-contrastive form \citep{ huszar2017variational, durkan2020contrastive}, where the classifier is still trained to learn the joint likelihood ratio using the ABC reference table (similarly as in B-GAN). However, the generator is now trained on $\Xn_0$ by maximizing the evidence lower bound. This algorithm is related to the B-GAN simulator, but uses $\Xn_0$ during the training stage.

Contrastive learning has been used inside Bayesian likelihood-free sampling algorithms before (see e.g. \citet{wang2021approximate,gutmann2018likelihood,kaji2021mh}).  
Both \citet{wang2021approximate} and \citet{kaji2021mh}  assume iid data  with a large enough sample size $n$ to be able to apply classification algorithms for each iteration  of Metropolis Hastings and ABC, respectively.
Our approach does not require iid data and can flexiblely accommodate almost any shape of data. We also do not require to run classification at each posterior simulation step.

We show highly competitive performance of our methods (relative to state-of-the art likelihood-free Bayesian methods) on several simulated examples.
While conceptually related methodology has occurred before  \citep{papamakarios2016fast, lueckmann2017flexible,ramesh2022gatsbi}, theory supporting these likelihood-free Bayesian approaches has been lacking. We provide new frequentist-Bayesian theoretical  results for the typical variational distance between the true and approximated posteriors. 
We analyze Wasserstein versions of both the B-GAN algorithm as well as adversarial variational Bayes. With properly tuned neural networks, we show that this distance goes to zero as $n\rightarrow\infty$ with large enough ABC reference tables.

  The outline of our paper is as follows. \Cref{sec:adv_bayes} reviews conditional GANs and introduces the Bayesian GAN sampler together with the reinforcement adjustments. In \Cref{sec:vb}, we describe another local enhancement strategy inspired by implicit variational Bayes. In \Cref{sec:theory}, we investigate the theoretical guarantees of the B-GAN posteriors. The performance of our methods is illustrated on simulated datasets  in \Cref{sec:performance} and a real data application in \Cref{sec:real_data}. In \Cref{sec:discussion}, we conclude with a discussion.

\section{Adversarial Bayes}\label{sec:adv_bayes}

Generative Adversarial Networks (GANs) \citep{goodfellow2016deep} are a game-theoretic construct in artificial intelligence  designed to simulate from likelihoods over complex objects. GANs involve two machines playing a game against one another. A {\em Generator} aims to deceive a {\em Discriminator} by simulating fake samples that resemble observed data while, at the same time, the  Discriminator learns to tell the fake and real data apart. This process iterates until the generated data are indistinguishable by the Discriminator and can be regarded as genuine likelihood samples.  Below, we review several recent GAN innovations and propose an incarnation for  simulation from a posterior as opposed to a likelihood.

\subsection{Vanilla GANs}\label{sec:vGAN}
In its simplest form, GANs learn how to implicitly simulate from the likelihood $p_{\theta_0}^{(n)}(\cdot)$ using only its realizations $\Xn_0\in\mX$ where $\Xn_0\sim p_{\theta_0}^{(n)}$ when $\theta_0$ is unknown.
 Recall that draws from implicit distributions can be obtained by passing a random noise vector $Z^{(n)}\in\mathcal Z$  through  a non-stochastic pushforward mapping  $g_{\b}(\cdot):\mathcal Z\rightarrow\mathcal X$. 
The original GANs formulation \citep{goodfellow2016deep}  involves a Generator, specified by the mapping $g_{\b}(\cdot)$, that attempts to generate samples similar to $\Xn_0$ by filtering $Z^{(n)}$, i.e. $Z^{(n)}\sim \pi_Z^{(n)}, \Xn=g_\beta(Z^{(n)}) \sim p_{\theta}^{(n)}$.

The  coefficients $\b$ in the Genrator are iteratively updated depending on the feedback received from the Discriminator. The Discriminator, specified by a mapping  $d(\cdot):\mathcal X\rightarrow(0,1)$, gauges similarity between 
$\Xn$ and $\Xn_0$ with a discrepancy between their (empirical) distributions.  {Hereafter we use $X$ to denote a generic dataset  as $X\in \mX$ for simplicity of notation.} At a population level,  a standard way of comparing  two distributions, say $P_{\theta_0}^{(n)}$ and $P_{\theta}^{(n)}$, is with the symmetrical Jensen-Shannon divergence 
which can be   written as a solution to a particular optimization problem 
\begin{equation}\label{eq:JS}
\text{JS}(P_{\theta}^{(n)},P_{\theta_0}^{(n)})=\ln 2+0.5\times\sup_{d:\mX\rightarrow (0,1)}\left\{E_{X\sim P_{\theta}^{(n)}}\ln \big[d(X)\big]+E_{X\sim P_{\theta_0}^{(n)}}\ln\big[1-d(X)\big]\right\}.
\end{equation}

The optimal Discriminator $d^*(\cdot)$, solving the optimization \eqref{eq:JS}, is $d^*(X)=p_\theta^{(n)}(X)/\big[p_\theta^{(n)}(X)+p_{\theta_0}^{(n)}(X)\big]$ \citep[Proposition 1]{goodfellow2014generative}. The optimal Generator is then defined through  the optimal value $\b^*$  which leaves the Discriminator  maximally confused, i.e. $d^*(X)=1/2$ and therefore $p_{\theta_0}^{(n)}(X)=p_\theta^{(n)}(X)$ uniformly over $\mathcal X$.

Despite the nice connection to likelihood ratios, original GANs \citep{goodfellow2014generative} may suffer from training difficulties when the discriminator becomes too proficient early on \citep{gulrajani2017improved, arjovsky2017towards}.   Alternative divergences have been implemented inside GANs that are less prone to these issues. For example,  the Wasserstein distance \citep{arjovsky2017towards} also admits a dual representation 
\begin{equation}\label{eq:wasser}
d_W\Big(P_{\theta}^{(n)},P_{\theta_0}^{(n)}\Big)= \sup_{f \in  \mF_W} \left|E_{X\sim P_{\theta}^{(n)}} f(X) - E_{X\sim P_{\theta_0}^{(n)}}f(X)\right|
\end{equation}
where  $\mF_W=\{f: \norm{f}_L\leq 1\}$ are functions with a Lipschitz semi-norm $\norm{f}_L$ at most one. The function $f(\cdot)$ is often referred to as the {\em Critic}. In our implementations, we will concentrate on the Wasserstein version of GANs \citep{arjovsky2017wasserstein}.

\subsection{Conditional GANs for Bayes}\label{sec:cond_GANs}
While originally intended for simulating from likelihoods underlying observed data,  GANs can be extended to  simulating from  distributions {\em conditional on} observed data.
Certain aspects of conditional GANs (cGANs) have been investigated earlier \citep{gauthier2014conditional, mirza2014conditional} in various contexts including causal inference \citep{athey2021using} or non-parametric regression \citep{zhou2022deep}. Our work situates conditional GANs firmly within the context of ABC and likelihood-free posterior simulation.  Before we describe our development in Section \ref{sec:bayesian_GANs}, we first introduce the terminology of cGANs within a  Bayesian context. We will intentionally denote  with $X$  the conditioning variables and focus on the conditional distribution $\pi(\theta\C X)$ for the inferential parameter $\theta\in\Theta$ with a prior $\pi(\theta)$.
Similarly as with vanilla GANs in Section \ref{sec:vGAN},  cGANs again involve two adversaries   represented by two mappings.  We focus on the Wasserstein version here. For readers that are less familiar with GANs and cGANs, we include the JS-based cGAN in \Cref{sec:JS_detail}, which is more intuitive but less stable in training.

\begin{definition}(Generator) \label{def:generator}
We define  a deterministic {\em generative model}   as  a mapping $g:(\mathcal Z\times \mathcal X)\rightarrow \Theta$ that filters  noise random variables $Z\in\mathcal Z$ to obtain samples from an implicit conditional density $\pi_g(\theta\C X)$. This conditional model then defines an implicit joint model $\pi_g(X,\theta)=\pi_g(\theta\C X)\pi(X)$, where $\pi(X)=\int_\Theta p_{\theta}^{(n)}(X)\pi(\theta)\d\theta$ is the marginal likelihood.
\end{definition}

\begin{definition}(Critic) 
	We define a deterministic {\em critic model} as a mapping $f:(\mathcal X\times \Theta)\rightarrow \R$, which estimates the Wasserstein distance between the data pairs $(X,\theta)$ from $\pi(X,\theta)$ and $\pi_g(X,\theta)$.
\end{definition}

The main distinguishing feature, compared to vanilla GANs, is that the conditioning random vector $X$ enters {\em both} mappings.
The task is to flexibly parametrize $g_{\b}(\cdot)$, e.g. using neural networks as will be seen later, in order to approximate the joint density model $\pi(X,\theta)$ as closely as possible.
Ideally, we would like to recover an (oracle) function $g^*: \mathcal Z \times \mX \to \Theta$ such that the conditional distribution of $g^*(Z,X)$ given $X$ is the same as  $\pi(\theta\mid X)$.  The existence of such an oracle $g^*$ is encouraged by the noise-outsourcing lemma from probability theory  \citep{kallenberg2002foundations,zhou2022deep} which we reiterate as Lemma \ref{lem:noise_outsourcing} in Section \ref{sec:JS_detail} for Gaussian $Z$.  
{Although it is not necessary for $\Theta$ and $Z$ to be of the same dimension $d$, we choose $\pi_Z= N(0, I_d)$ to balance the expressiveness of the generator and the discriminator.}

The premise of conditional GANs rests in the fact that matching two joint distributions, while fixing a marginal distribution, is equivalent to matching conditional distributions. This implies  that $g^*(Z, X)$, given $X$, is indeed distributed according to  $\pi(\theta\C X)$. The question remains how to find the oracle mapping $g^*$ in practice. In the Wasserstein cGAN, the minimax game is characterized as 
\begin{equation}\label{eq:wcgan}
(g^*,f^*)=\arg\min\limits_{g \in \mG}\max\limits_{f \in \mF} \left|  E_{X\sim \pi(X), Z\sim \pi_Z} f(X,g(Z,X))- E_{(\theta, X)\sim \pi(X,\theta)}   f(X,\theta)\right|,
\end{equation}
where, using  $\mF=\mF_W$ (1-Lipschitz in $\theta$), $g^*$ minimizes the Wasserstein distance between $\pi_g(X,\theta)$ and $\pi(X,\theta)$. When $\mG$ and $\mF$ are expressive enough, the minimum is achieved if and only if the Wasserstein distance between $\pi_{g^*} (X, \theta)$ and $\pi(X, \theta)$ is zero, which equivalently means $\pi_{g^*} (X, \theta)=\pi(X, \theta)$. With the same marginal $\pi(X)$, the solution $g^*$ essentially satisfies 
$\pi_{g^*} (\theta\mid X )= \pi(\theta\mid X)$.

\subsection{Generative Bayesian Inference with GANs}\label{sec:bayesian_GANs}

To implement the adversarial game \eqref{eq:wcgan} in practice, one needs to (a) parametrize $\mathcal F$ and $\mathcal G$ (for instance using neural networks) and (b) to replace the expectations in \eqref{eq:wcgan} with empirical counterparts.
 Both of these steps will introduce approximation error. We provide theoretical insights later in Section \ref{sec:theory}.
We assume that the generator class $\mG=\{g_{\b}:(\mathcal Z\times \mathcal X)\rightarrow \Theta\,\,\text{where}\,\, \b\in \R^G\}$ is parametrized with $\b$   and the critic class $\mF=\{f_{\bomega}:(\mathcal X\times \Theta)\rightarrow \R\,\,\text{where}\,\, \bomega\in \R^C\}$ is parametrized with $\bomega$. We use neural networks (with ReLU activations)   and support this choice with  theory (Corollary \ref{cor:dnn}).

For the empirical version, one can  use the ABC reference table which consists of  simulated data pairs $\{(\theta_j,\Xn_j)\}_{j=1}^T$ 
generated from the joint  model $\pi(\Xn,\theta)=p_{\theta}^{(n)}(\Xn)\pi(\theta)$ under the prior $\pi(\theta)$.  
 Specifically, we use $n$ to refer the total data dimensionality. Each draw $\Xn_j\sim P^{(n)}_{\theta}$ can be either $n^*$ iid $q$-dimensional observations, or $q$-dimensional times-series of length $n^*$, such that $n=n^**q$.
From now on, we will simply denote $\Xn_j$ with $X_j$, similarly for $\Xn_0$ and $\Xn$.

We can break the relationship between $\theta$ and $X$ in the ABC reference table by contrasting these data pairs  with another dataset consisting of $\{(g(Z_j,X_j),X_j)\}_{j=1}^T$ where $Z_j$'s have been sampled from $\pi_Z(\cdot)$. 
Keeping the same $X_j$'s essentially means that we are keeping the same marginal.   The dataset $\{(g(Z_j,X_j),X_j)\}_{j=1}^T$  encapsulates iid  draws from $\pi_g(X,\theta)$. These two datasets can be then used to approximate the expectations in \eqref{eq:gan_obj} and \eqref{eq:wcgan}.
A high-level description of an algorithm for solving  the   Jensen-Shannon (JS) version of the game from Lemma \ref{lemma:cGAN} is outlined in \Cref{alg:KL-GAN}
(Appendix Section \ref{sec:JS_cGAN}).  
 As mentioned earlier, the JS version may suffer from training issues \citep{arjovsky2017towards}.  We provide an illustration of such  issues using convergence diagnostics on a toy example  in Section \ref{sec:JS_cGAN}.

\begin{figure}[!t]
\centering
\begin{minipage}{\linewidth}
\begin{algorithm}[H]
	\caption{\em  Vanilla B-GAN   }\label{alg:WGAN}
\centering
	
\resizebox{0.8\linewidth}{!}{
		\begin{tabular}{l l}
\hline
			\multicolumn{2}{c}{\bf Input \cellcolor[gray]{0.6} }\\	
				\hline
			\multicolumn{2}{c}{Prior $\pi(\theta)$, observed data $X_0$ and noise distribution $\pi_Z(\cdot)$}\\
				\hline
			\multicolumn{2}{c}{\bf Training \cellcolor[gray]{0.6}}\\		
			\hline
			\multicolumn{2}{c}{Initialize  network parameters {$\bomega^{(0)}=0$ and $\b^{(0)}=0$}}\\
	
			\multicolumn{2}{c}{  \cellcolor[gray]{0.9}{\bf Reference Table} }\\
			\multicolumn{2}{l}{ For $j=1,\dots, T$:\qquad Generate   $(X_j,\theta_j)$ where $\theta_j\sim\pi(\theta)$ and $X_j \sim P_{\theta_j}^{(n)}$. } \\
			\multicolumn{2}{c}{ \cellcolor[gray]{0.9}{\bf Wasserstein GAN}  }\\
			\multicolumn{2}{l}{For $t=1,\dots, N$:}\\
			\quad {\bf Critic Update} ($N_\text{critic}$ steps): For $k=1, \ldots, N_\text{critic}$ \\
			\qquad \qquad  Generate   $Z_{j}\sim\pi_Z(z)$ for  $j=1,\ldots, T$.\\
			\qquad \qquad  Update $ \bomega^{(t)}$ by applying stochastic gradient descent on     \eqref{eq:wgan_obj}. &\\ 
			\quad {\bf Generator Update} (single step) \\
			\qquad  \qquad  Generate noise $Z_j\sim\pi_Z(z)$ for  $j=1,\ldots, N$. \\
			\qquad  \qquad Update $\b^{(t)}$ by applying stochastic gradient descent on    \eqref{eq:wgan_obj}. &\\
			\multicolumn{2}{c}{ \cellcolor[gray]{0.6}{\bf Posterior Simulation}: }\\
			\multicolumn{2}{l}{ For $i=1,\dots, M$:\qquad	Simulate	 $Z_i\sim\pi_Z(z)$ and set $\wt \theta_i=  g_{\b^{(N)}}(Z_i, X_0)$.}
				\end{tabular}}
\end{algorithm}
\end{minipage}
\end{figure}

In our implementations, we thereby consider the empirical version of  \eqref{eq:wcgan} which again involves simulated datasets $\{(\theta_j,X_j)\}_{j=1}^T$  and   
$\{Z_j\}_{j=1}^T \iid \pi_Z(\cdot)$  to obtain
\begin{equation}\label{eq:wgan_obj}
\wh\b_T=\arg\min\limits_{\b: g_{\b}\in \mG}\left[\max\limits_{\bomega: f_{\bomega}\in \mF_W}\left|\sum_{j=1}^T f_{\bomega}\big(X_j,g_{\b}(Z_j,X_j)\big) - \sum_{j=1}^T f_{\bomega}(X_j,\theta_j)\right|\right].
\end{equation}
One particular way of solving this problem is summarized in \Cref{alg:WGAN}. 
In terms of the constraint on $\bomega$ to ensure the Lipschitz condition   $\|f_{\bomega}\|_L\leq 1$,
 the original  Wasserstein GANs implementation \citep{arjovsky2017towards} used gradient clipping, which may lead to computational issues.  Alternatively, \citet{gulrajani2017improved}   imposed a soft version of the constraint with a penalty
 on the  gradient of  $f_{\bomega}$  with respect to a convex combination of the two contrasting datasets. We adopt the one-sided penalty same as \citet{athey2021using}, with details in \Cref{sec:implementation}. To stabilize gradients,  the critic is updated multiple times (ideally until convergence) before each update of the generator, which is different from  \Cref{alg:KL-GAN} where such a stabilization may not be feasible  \citep{arjovsky2017wasserstein}.

Our  GAN  framework for Bayesian posterior simulation (i.e. \Cref{alg:WGAN} further referred to as B-GAN) consists of a neural sampler $g_{\wh\b_T}(Z,X)$ which generates samples from an approximate posterior, i.e. conditioning on the observed data $X_0$, by filtering iid noise as follows
\begin{equation}\label{eq:simulate1}
\wt \theta_j=g_{\wh\b_T}(Z_j,X_0)\quad\text{where}\quad Z_j\iid\pi_Z(\cdot)\quad\text{for}\quad j=1,\dots,M.
\end{equation}
When $g_{\wh\b_T}$ is close to $g^*$, the samples $\wt\theta_j$ will arrive approximately from $\pi(\theta\C X_0)$. One of the practical appeals of this sampling procedure is that, once the generator has been trained, the simulation cost is negligible.
Note that our observed data $X_0\sim p_{\theta_0}^{(n)}(X)$ are {\em not involved} in the training stage, {\em only} in the simulation stage \eqref{eq:simulate1}. 
We illustrate \Cref{alg:WGAN} on a toy example.  The configurations of our B-GAN networks and optimization hyperparameters  are described in  \Cref{sec:implementation_toy}.  Note that our B-GAN approach is different from the Bayesian GAN of \citet{saatci2017bayesian}, which places a probabilistic structure over the GAN parameters and update the parameteters via sampling instead of pure optimization. Our work utilizes the GAN framework for likelihood-free Bayesian inference and the training scheme is optimization-based.

\begin{figure}[!t]
\centering
\includegraphics[width=0.95\textwidth]{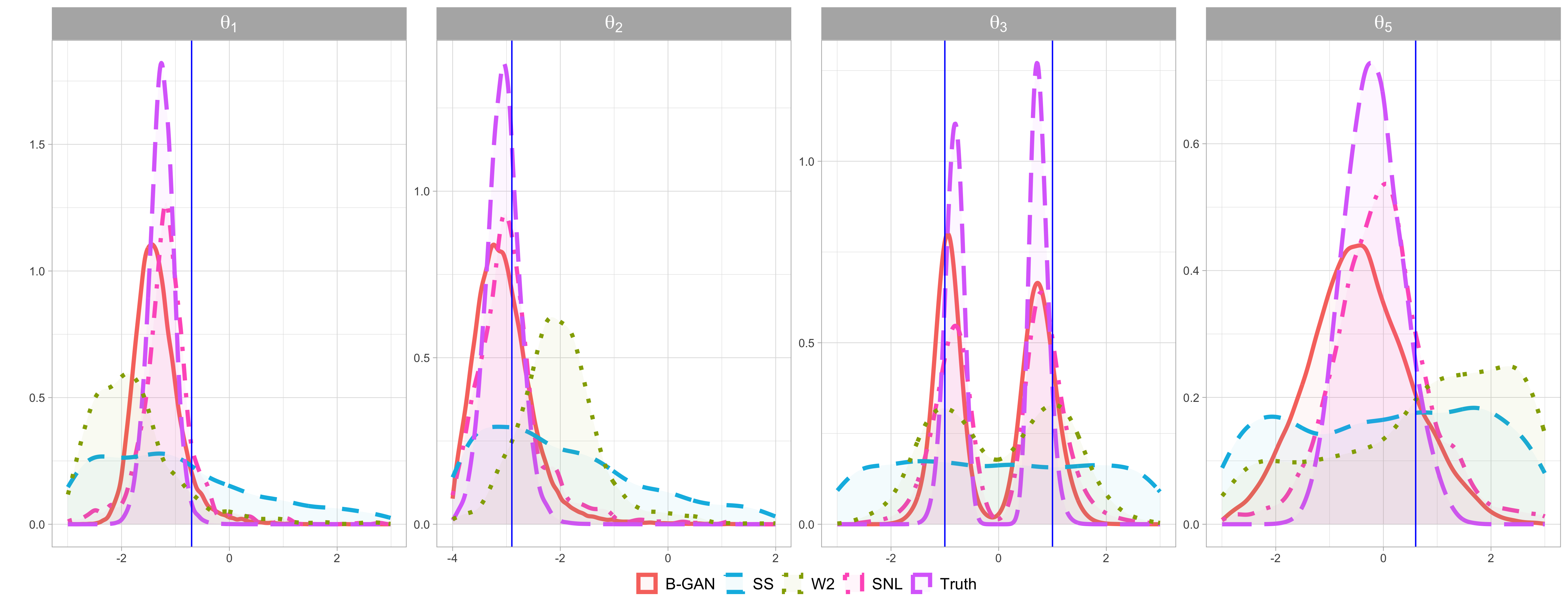}
\caption{\footnotesize The approximate d posteriors given by B-GAN, SNL, SS, and W2 for the toy example. The results for $\theta_4$ are similar to $\theta_3$ and thus not shown here. }\label{fig:gauss_ex1}
\end{figure}

\begin{example}[Toy Example]\label{ex:toy}
This toy example  (analyzed earlier in  \citep{papamakarios2019sequential}) exposes the fragility of ABC methods in a relatively simple setting. 
The experiment entails  $n=4$  two-dimensional Gaussian observations  
$X =(x_{1}, x_{2}, x_{3}, x_{4})'$ with $x_j\sim \mN(\mu_{\btheta}, \Sigma_{\btheta})$ parametrized by $\btheta =(\theta_1,\theta_2,\theta_3,\theta_4,\theta_5)'$, where
\vspace{-0.5cm}
 $$ \mu_{\btheta}=(\theta_1, \theta_2)'  \quad \text{and}\quad  \Sigma_{\btheta}= \left(\begin{matrix}
 s_1^2 &\rho s_1 s_2\\
\rho s_1 s_2  &  s_2^2\\
 \end{matrix}
 \right)
 $$
with $s_1=\theta_3^2,  s_2=\theta_4^2$ and  $\rho=\text{tanh}(\theta_5)$. The parameters $\btheta$ are endowed with a uniform prior on $[-3,3]\times[-4,4]\times [-3,3]\times [-3,3]\times [-3,3]$. Approximating the posterior can be tricky because the signs of parameters $\theta_3$ and $\theta_4$ are not  identifiable, yielding multimodality.
We generate  $X_0$ with parameters $ \btheta_0 = (-0.7, -2.9, -1.0, -0.9, 0.6)'$. Since we have access to the true posterior, we can directly compare our posterior reconstructions with the truth. 
We compare B-GAN with (1) ABC using naive summary statistics (SS) (mean and variance), (2) 2-Wasserstein distance ABC \citep{bernton2019approximate}, and (3) Sequential Neural Likelihood (SNL) \citep{papamakarios2019sequential} with the default setting suggested by the authors. We provide all implementation details in \Cref{sec:implementation_toy}. For each method, we obtained $M=1\,000$ samples and plotted them in \Cref{fig:gauss_ex1}. Our B-GAN approach as well as SNL nicely capture the multimodality.
 Although the B-GAN posterior quickly locates a neighborhood around true values $\theta_0$, its variance appears to be bigger than the true posterior, suggesting an approximation gap. This is expected because B-GAN is trained to perform well on average for any $X$, {\em not} necessarily for $X_0$. This motivates  our two refinement strategies: one based on active learning (Section \ref{sec:gan_ri}) and one based on variational Bayes (Section \ref{sec:vb}). 
\end{example}
Similarly as  with default ABC techniques, our B-GAN approach is not query-efficient, i.e. many prior guesses $\theta_j$ in the training dataset may be too far from the interesting areas with a likelihood support leaving only a few observations to learn about the conditional $\pi(\theta\C X_0)$. The next section presents a two-step approach which uses $X_0$ for proposal construction in the ABC reference table to obtain more valuable data-points in the reference table.

\subsection{Two-step Refinement} \label{sec:gan_ri}

\begin{figure}[!t]
\centering
\begin{minipage}{.8\linewidth}
\begin{algorithm}[H]
	\caption{\em  2-Step Refinement (B-GAN 2step) }\label{alg:WGAN-RL}
\centering
	\resizebox{\linewidth}{!}{	
			\begin{tabular}{l l}
\hline
			\multicolumn{2}{c}{\bf INPUT \cellcolor[gray]{0.6} }\\	
				\hline
			\multicolumn{2}{c}{Prior $\pi(\theta)$, observed data $X_0$ and noise distribution $\pi_Z(z)$}\\
					\hline
			\multicolumn{2}{c}{\bf Training \cellcolor[gray]{0.6}}\\		
			\hline
			\multicolumn{2}{c}{Initialize  network parameters {$\bomega^{(0)}=0$ and $\b^{(0)}=0$}}\\

			\multicolumn{2}{c}{\bf Pilot Run \cellcolor[gray]{0.9}}\\		
			\multicolumn{2}{c}{Apply  \Cref{alg:WGAN} with $\pi(\theta)$ to learn $\wh g_{pilot}(\cdot)$ }\\
			\multicolumn{2}{c}{\bf  \cellcolor[gray]{0.9}Reference Table}\\
				\qquad  Generate  pairs $\{(X_j,\theta_j)\}_{j=1}^T$ where $\theta_j=\wh g_{pilot}(Z_j,X_0)$ for $Z_j\sim\pi_Z$ and $X_j \sim P_{\theta_j}^{(n)}$.  &\\

			\multicolumn{2}{c}{\bf  \cellcolor[gray]{0.9}Refinement}\\
			
				\multicolumn{2}{c}{Apply Wasserstein GAN  step in \Cref{alg:WGAN} on $\{(X_j,\theta_j)\}_{j=1}^T$ and return $g_{\wh\b_T}(\cdot)$}\\
				
			\multicolumn{2}{c}{\bf  \cellcolor[gray]{0.9}Posterior Simulation}\\
			\qquad Simulate	 $\{Z_i\}_{i=1}^M\iid\pi_Z(z)$ and set $\wt \theta_i=  g_{\wh \b_T}(Z_i, X_0)$.\\
			\qquad Estimate $\hat w_i$ using either \eqref{eq:kde_plugin} or \eqref{eq:weights}.\\
	\hline
		\multicolumn{2}{c}{\bf OUTPUT \cellcolor[gray]{0.6}}\\			
\hline
		\multicolumn{2}{c}{ Pairs of posterior samples and weights $(\wt \theta_1,\hat w_1), \ldots, (\wt \theta_M, \hat w_M)$ }\\
				\end{tabular}}
\end{algorithm}
\end{minipage}
\end{figure}

Our chief goal is to find a high-quality approximation to the conditional distribution $\pi(\theta\C X)$ evaluated at the observed data $X=X_0$, not necessarily uniformly over the entire domain $\mathcal X$.
However, the ABC reference table $\{(\theta_j,X_j)\}_{j=1}^T$  may not contain enough data points $X_j$ in the vicinity of $X_0$ to train the simulator when the prior $\pi(\theta)$ is too vague.  
This can be remedied by generating a reference table using an auxiliary proposal distribution $\wt\pi(\theta)$ which is more likely to produce pseudo-observations $X_j$ that are closer to $X_0$. For example, a pilot simulator $g_{\wh \b_T}(Z,X_0)$ in \eqref{eq:simulate1} obtained from \Cref{alg:WGAN} under the original prior $\pi(\theta)$ can be  used to guide  simulations in the next round to  sharpen the reconstruction accuracy around $X_0$ \citep{papamakarios2016fast}.
Training the generator  $\wt g_{\wh \b_T}$, simulating from an approximation to $\wt\pi(\theta\C X_0)\propto p_{\theta_0}^{(n)}(X_0)\wt\pi(\theta)$, under the `wrong' prior  can be corrected for by  importance re-weighting with weights 
$r(\theta)=\pi(\theta)/\wt\pi(\theta)$.  Since the posterior $\wt\pi(\theta\C X_0)r(\theta)$ is proportional to $\pi(\theta\C X_0)$,
reweighting  the resulting samples $\wt \theta_j=\wt g_{\wh \b_T}(Z,\Xn_0)$  with weights  $w_j=r(\theta_j)$ will produce samples from an approximation to the original posterior (after normalization).
Algorithm \ref{alg:WGAN-RL} summarizes this two-step strategy, referred to as B-GAN-2S.

Since the  proposal density $\pi_{g_{\wh\b_T}}(\theta\C X_0)$ obtained in the pilot run {\em may not} have an analytical form, computing the importance weights $w_j=\pi(\theta_j)/\pi_{g_{\wh\b_T}}(\theta_j\C X_0)$ directly may not feasible. 
The density ratio $r(\theta)$ can  be approximated, however. For example, with a tractable prior $\pi(\theta)$ the importance weights $w_j$ can be estimated by
\begin{equation}\label{eq:kde_plugin}
\hat w_j= \frac{\pi(\theta_j)}{\hat \pi_{g_{\wh\b_T}}(\theta_j\C X_0)}
\end{equation}
where  $\hat \pi_{g_{\wh\b^T}}(\theta_j\C X_0)$ is a plugged-in kernel density estimator (KDE) estimator  \citep{terrell1992variable}.  This is particularly useful and efficient when the parameter dimension  is low. 
When the prior is also not tractable but simulatable,  the weights $w_j$ can be  directly estimated from classification by contrasting datasets $\wt\theta_j\sim\wt\pi(\theta)$ (label `0') and $\theta_j\sim\pi(\theta)$ (label `1'). 
In particular, training a classifier $\tilde D$ (see e.g.  \citet{cranmer2015approximating, durkan2020contrastive, gutmann2012noise}  for the explanation of the 'likelihood-ratio-trick' for classification based estimators), we can obtain
\begin{equation}\label{eq:weights}
\hat w_j=\frac{\tilde D(\wt \theta_j)}{1-\tilde D(\wt \theta_j)}. 
\end{equation}

\citet{papamakarios2016fast} used mixture density networks estimators of the conditional distribution $\pi(\theta\C x)$ after a pilot run to learn the proposal distribution $\wt \pi(\theta)$. In order to obtain an analytically tractable Gaussian mixture representation, their proposal $\wt\pi(\theta)$ has to be Gaussian and it cannot be narrower than any of the mixture components. We do not require such assumptions.
\citet{lueckmann2019likelihood} instead propose to directly incorporate the weights $r(\theta)$ inside training and relax the Gaussianity assumption to avoid the variance instability.  
Similarly as in \citet{lueckmann2019likelihood}, we could also incorporate weights $\hat w_j$ inside the objective function, e.g. multiplying each summand in \eqref{eq:wgan_obj} by $\hat w_j$. 

In the two-step refinement, the observed data $X_0$ only contribute to the proposal distribution $\wt\pi(\theta)$, not the training of the simulator $g_{\wh\b}(\cdot)$.
In  Section \ref{sec:vb}, we consider a variational Bayes variant which does involve $X_0$ in training. 

\begin{figure}[!t] 
\centering
\includegraphics[width=0.95\textwidth]{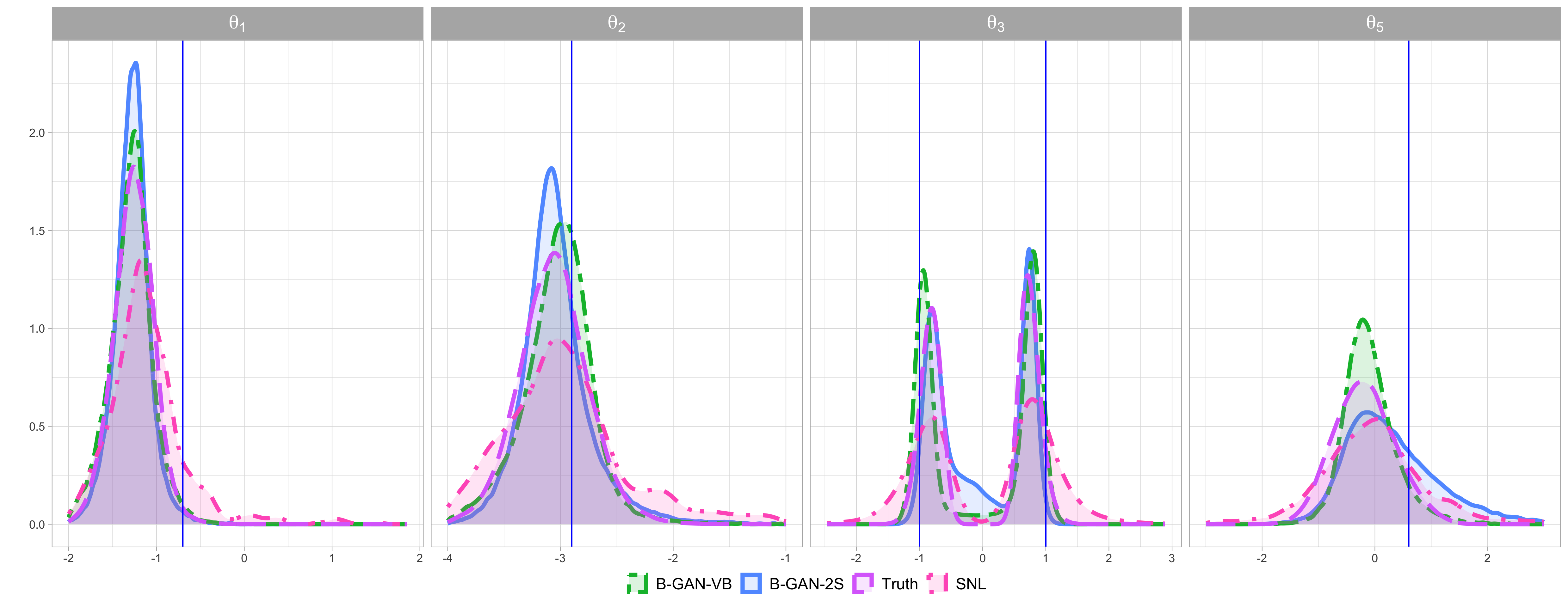}
\caption{\small Posterior densities under the Gaussian model. The true parameter is $\btheta_0=(-0.7, -2.9, -1.0, -0.9, 0.6)'$, while the signs of $\theta_3$ and $\theta_4$ are not identifiable.}\label{fig:gauss_ex2}
\end{figure}

\begin{example}[Toy Example Continued]\label{ex:toy_RL}
We continue our exploration of the Toy  Example \ref{ex:toy}.  We now use the output from B-GAN  (\Cref{alg:WGAN}) under the original uniform prior as a proposal distribution $\wt\pi(\theta)$ and generate training data  $\{(X_j,\theta_j)\}_{j=1}^T$ with a marginal $\wt\pi(X)=\int_{\mX}p_{\theta}^{(n)}(X)\wt\pi(\theta)d\theta$. To revert the generated posterior samples to the original uniform prior, we perform reweighting by $r(\theta)=\pi(\theta)/{\wt\pi(\theta)}$ using the kernel density estimator of $\wt\pi(\theta)$ obtained from the pilot B-GAN run in {\Cref{alg:WGAN}.  The number of training points used in the second step is $T=50\,000$.} We note that much smaller $T$ could be used if one were to perform more sequential refinements, not just one.  The re-weighted (and normalized) posterior is plotted against the truth, SNL and a variational Bayesian variant (introduced later in Section \ref{sec:vb}) in \Cref{fig:gauss_ex2}. Both B-GAN and B-GAN-2S provide tighter approximations to the true posterior  We repeated the experiment $10$ times and report the Maximum Mean Discrepancies (MMD) \citep{gretton2012kernel} between the true posterior and its approximations obtained from {$M=1\,000$ posterior draws} for each method in \Cref{fig:gauss_mmd}. Satisfyingly, B-GAN-2S yields the smallest MMD. We support this encouraging finding with our theoretical results in \Cref{sec:theory}. 

\begin{figure}[!t]
\centering
\includegraphics[width=0.8\textwidth, height=0.25\textwidth]{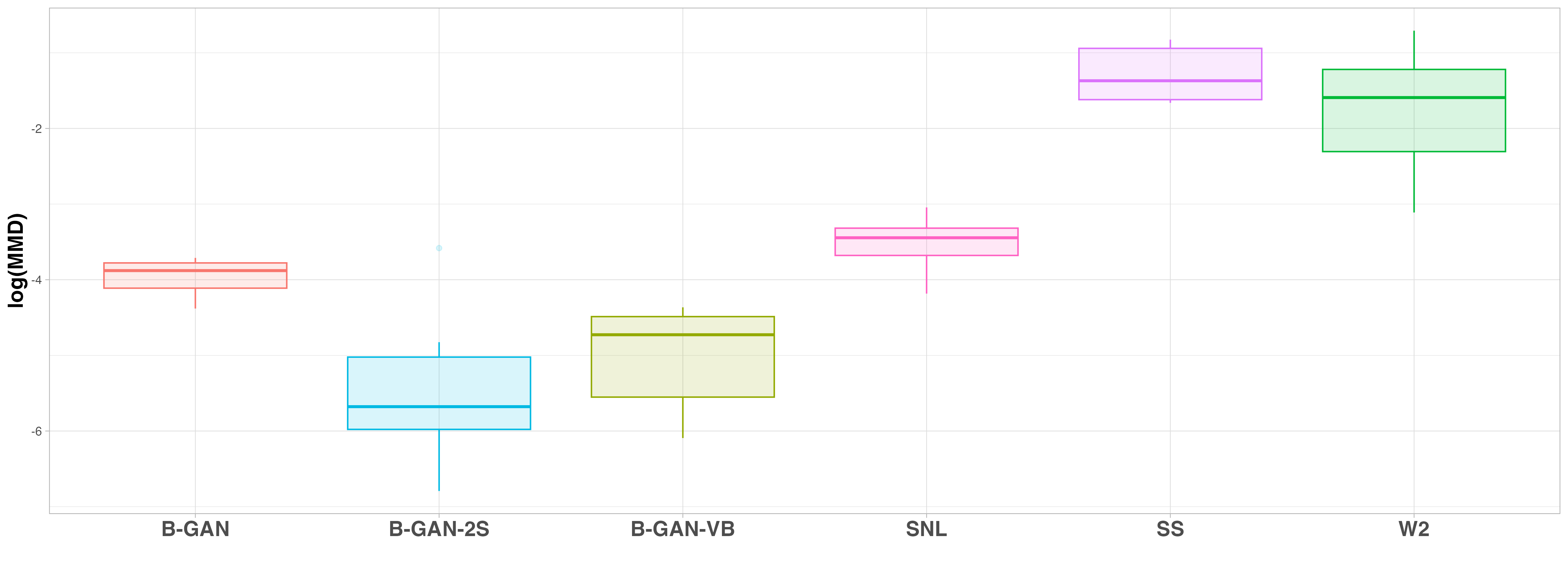}
\caption{\footnotesize Maximum Mean Discrepancies (MMD,  log scale) between the true posteriors and the approximated posteriors. The box-plots are computed from 10 repetitions. }\label{fig:gauss_mmd}
\end{figure}

\end{example}

\subsection{The Case of i.i.d. Observations}\label{sec:iid}

In this subsection, we begin by reviewing three key challenges that arise when the total data dimensionality 
$n$ is large. We then demonstrate how these challenges can be mitigated more effectively when the data consists of i.i.d. observations with large $n^*$ but moderate $q$. 

The first challenge stems from the fact that a large $n$ demands more complex and  expressive neural networks, which are often harder to train and require generating even bigger ABC reference tables. Second, as $n^*$ increases, the posterior becomes more concentrated, making it difficult to efficiently identify high posterior density regions starting from a non-informative prior. Lastly, the storage and memory requirements grow substantially when simulating or processing the same number of ABC draws with larger $n^*$.

The most straightforward and naive implementation is to flatten the $n^*\times q$ data matrix into a vector and condition on the entire vector in our GAN sampler. However, this approach inherits all three challenges and should be avoided when $n=n^**q$ is large.  Instead, we propose two alternative implementations tailored to i.i.d. datasets with large $n^*$, which help alleviate some of these computational burdens.

The first solution is to split the dataset into multiple batches, and use sequential update to parse only one batch at a time. For example, for i.i.d. dataset with very large $n^*$, we partition the data into batches of size $B$, parse the data into vectors of size $B\times q$, then use the posterior from the previous batch as a prior for the next. This is coherent with the Bayesian framework. The neural network can also be recycled from previous batch, which should further relieve the computational burden. This can also be applied for some dependent datasets, such as the Lotka-Volterra model, whose dependency is solely determined by the observation from last timestamp. For this type of model, sequential parsing is still valid with partition along the time horizon. For this solution, it helps with second challenge and the third challenge, and can mildly relieve the first challenge. The sequential update allows us to gradually approach the concentrated high posterior density region and reduces the size of datasets that needs to be generated each time to be $T*B*q$ instead of $T*n^**q$. Since the length of the flatten data vector is lower than $n^**q$, it requires less complex network structrues compared to the naive implementation, but this can still quickly grow overwhelming when $B$ also becomes really large. 

 The second solution leverages the exchangeable nature of i.i.d. observations. For an ideal parameter-efficient generator network, it should output the same conditional distribution regardless of the order of exchangeable inputs. In another words, it should satisfy
 \begin{align*}
 g_\beta(Z, [X_1, X_2, \ldots, X_n]) =g_\beta(Z, [X_{\pi(1)},X_{\pi(2)}, \ldots,X_{\pi(n)} ])
 \end{align*}
where $\pi(\cdot)$ can be any permutation. Similarly for the critic function, it should be able to provide the same Wasserstein distance estimate under any permutation of the conditional inputs.
 
For this parameter-efficient network, we adopt the deep set architecture from Zaheer et al. (2017), whose output is invariant to the order of inputs. To accommodate for the  dependency between the parameter and each individual observation, we parametrize the generator and critic networks as 
\begin{align*}
g_\beta(Z, [X_1, X_2, \ldots, X_n]) &= g^{(1)}_{\beta_1} (Z, \sum_{i=1}^n g^{(2)}_{\beta_2}(Z, X_i))\\
f_\omega(\theta, [X_1, X_2, \ldots, X_n]) &=  f^{(1)}_{\omega_1} (\theta, \sum_{i=1}^n f^{(2)}_{\omega_2}(\theta, X_i))
\end{align*}
where $g^{(1)}_{\beta_1},  g^{(2)}_{\beta_2}, f^{(1)}_{\omega_1}, f^{(2)}_{\omega_2}$ are sub-network structures inside the generator and critic functions. These structures allow us to efficiently learn conditional distribution on an exchangeable dataset without exploding network complexity.

We illustrate the efficiency of the three different implementations on the benchmark SBI example, M/G/1-queuing model. We follow the setting in \citet{fearnhead2011constructing}. Each observation is a 5-dimensional vector consisting of the first five inter-departure times $x_i = (x_{i1}, x_{i2}, x_{i3}, x_{i4}, x_{i5})'$. In the model, the service times $u_{ik}$ follow a uniform distribution $U[\theta_1, \theta_2]$, and the arrival times $w_{ik}$ are exponentially distributed with the rate $\theta_3$. We only observe the interdeparture times $x_i$, given by the process $x_{ik}=u_{ik}+\max(0, \sum_{j=1}^k w_{ij} - \sum_{j=1}^{k-1}x_{ij})$. We set true parameter to be $\theta_0=(1,5,0.2)$ and the prior on $(\theta_1, \theta_2-\theta_1, \theta_3)$ to be uniform on $[0,10]^2\times[0, 0.5]$.

To show how different implementations scale with $n$, we show three examples with $n=50, 100, 200$ in \Cref{fig:mg1}. Here we consider three treatments: (1) deep set structure (deep set), (2) sequential update by splitting $n^*$ observations into $5$ batches (sequential), and (3) stack all $n$ observations into one long vector (stacked). To show how the performance of these implementations vary with sample size $n^*$, each implementation uses the same network structure, i.e., same network complexity across different $n^*$. Details of the implementations are provided in \Cref{app:iid_implementation}.

\begin{figure}[!ht]
\centering
\includegraphics[width=0.75\textwidth]{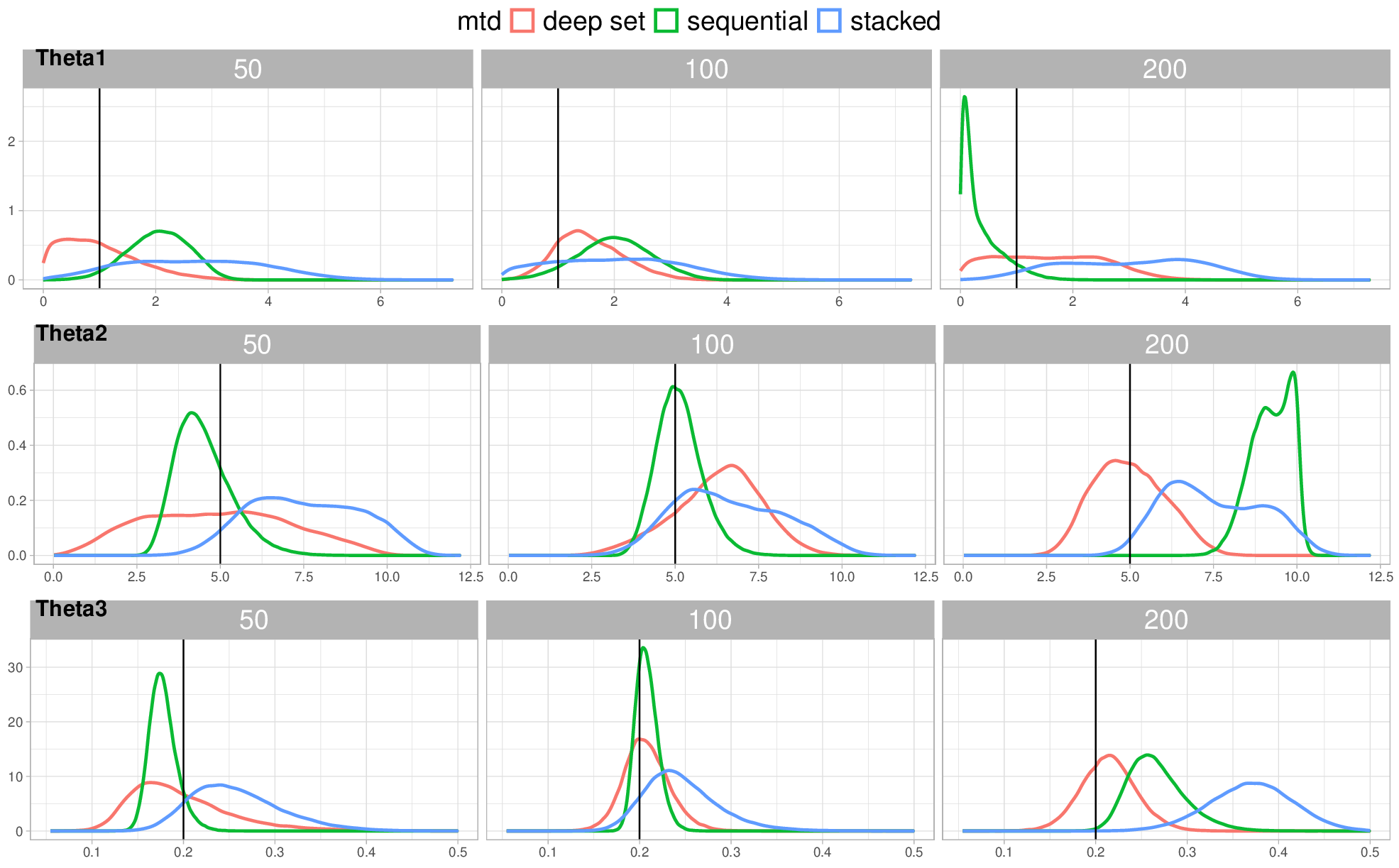}
\caption{Posterior for M/G/1-queuing under different implementations of B-GAN}\label{fig:mg1}
\end{figure}

From \Cref{fig:mg1}, we observe that for the deep set structure, the approximated posterior for $\theta_2$ and $\theta_3$ gets more and more concentrated when $n^*$ increases, and the approximated posterior covers the true values in all three $n^*$, showing that using the parameter-efficient network structure can scale better with $n^*$ for i.i.d. datasets. For the simple stacking implementation, since the input dimension becomes $5n^*$, it suffers from all three challenges we mention earlier. We observe that the network is able to give some rough approximation when $n^*=50$ and $n^*=100$ but fails to learn when $n^*=200$. For the batch sequential update, the input dimension is $5n^*/5=n^*$, so it is less affected by the growing sample size compared with the stacking implementation. Since the batch updates also help gradually learning the high posterior density area, the batch update implementation produces the most concentrated posterior for $n^*=50$ and $n^*=100$. However, it also fails to learn when $n^*=200$. This could happen due to  the cumulated errors in the batch update.

For the deep set structure, it largely circumvent the first challenge of exploding network complexity and partially relieve the second challenge from incorporating the statistical knowledge of exchangeability into the design. It cannot avoid the last challenge of growing data size which could be overwhelming for users with limited storage or memory. We recommend combing the deep set structure with  the sequential update for extremely large $n$ to achieve maximum efficiency.

\section{GAN Variational Bayes}\label{sec:vb}

Variational Bayes (VB) is an optimization-centric Bayesian inferential framework based on minimizing a divergence between approximate and real posteriors.
 VB typically reduces the infinite-dimensional optimization problem to a finite-dimensional one by forcing approximations into structured parametric  forms. 
Implicit distributions (defined as probabilistic programs) have the potential to yield  finer and tighter VB posterior approximations \citep{huszar2017variational, kingma2013auto, tran2017hierarchical, titsias2019unbiased}. This section highlights the connection between the  implicit variational Bayes inference and our B-GAN framework (Algorithm \ref{alg:WGAN} and \ref{alg:WGAN-RL}),  both of which target the posterior.

The VB setup consists of an (intractable) likelihood  $p_\theta^{(n)}(\cdot)$, prior $\pi(\theta)$  and a class of posterior approximations $q_{\b}(\theta\C X_0)$ indexed by $\b$. We are recycling the notation of $\b$ here to highlight the connection between the GAN generator and  the implicit variational  generator.
The goal of the VB approach is to find a set of parameters $\b^*$ that maximize  the evidence lower bound (ELBO) to the marginal likelihood 
\begin{align}\label{eq:elbo1}
\log \pi(X_0) \geq \mathcal L(\b)\equiv \int\log\left(\frac{\pi(X_0,\theta)}{q_{\b}(\theta\C X_0)}\right)q_{\b}(\theta\C X_0)\d\theta.
\end{align}
The tightness of the inequality (for the best approximating model within the class) increases with expressiveness of the inference model $q_{\b}(\cdot)$, where the equality occurs when 
 $q_{\b}(\theta\C X_0)=\pi(\theta\C X_0)$. 
Writing the ELBO $\mL(\b)= -\text{KL}\left( q_{\b}(\theta\C X_0)\| \pi(\theta\mid X_0)\right)+C$ in terms of Kullback-Leibler discrepancy, we have

\begin{equation}\label{eq:psistar}
\b^*=\arg\max_{\b} \mL(\b)=\arg\min E_{\theta \sim q_\beta(\theta\mid X_0)} \log\left(\frac{\pi(\theta\mid X_0)}{q_\beta(\theta\mid X_0)} \right)
\end{equation}

We estimate this expectation using $K$ Monte Carlo draws $\theta_k \sim q_\beta(\theta\mid X_0), k=1, \ldots, K$, yielding the empirical objective
\begin{equation}\label{eq:beta_VB}
\hat\b_\text{VB}=\arg\min_{\b} \frac{1}{K} \sum_{\theta_k \sim  q_\beta(\theta\mid X_0)} \log\left(\frac{\pi(\theta_k\mid X_0)}{ q_\beta(\theta_k\mid X_0)}\right)
\end{equation}
where  $q_\beta(\theta\mid X_0)$ is generated implicitly via $Z_k \sim \pi_Z$ and $ \theta_k=g(Z_k, X_0)$.

However, contrary to conventional VB algorithms, we cannot directly evaluate \eqref{eq:beta_VB} since neither $g_\b(\theta\mid X_0)$ nor $\pi(\theta\mid X_0)$ are tractable. To address this, we leverage  contrastive learning \citep{bickel2007discriminative} to estimate the density ratio between two implicit distributions. Here we first show how the loss function in \eqref{eq:beta_VB} can be re-written with density ratio estimator, and then we delve into details on how we actually implement that due to training stability concerns.

Recall the ``oracle classifier'' $d^*_{g_\b}$ in \Cref{lemma:cGAN}, we have
\[
\frac{d^*_{g_\b}(X,\theta)}{1-d^*_{g_\b}(X,\theta)}=\frac{\pi(X,\theta)}{q_{\b}(\theta\C X)\pi(X)},
\]
which allows  us to re-write the VB loss in \eqref{eq:beta_VB} as
\begin{equation}\label{eq:beta_VB_oracle}
\hat\b_\text{VB}=\arg\min_{\b} \frac{1}{K} \sum_{k=1}^K \text{logit} (d^*_{g_\b}(X_0, \theta_k)).
\end{equation}
Since the ``oracle classifier'' $d^*_{g_\b}$ is unavilable, it has to be estimated by solving a classification problem using a particular class of classifiers $\mathcal D=\{d_{\bphi}:(\mX\times\Omega)\rightarrow (0,1); \bphi\in\R^C\}$, for instance, neural networks parameterized by $\phi$. This adversarial idea -- replacing components of the ELBO with learned discriminators -- has been explored in earlier works \citep{mescheder2017adversarial,huszar2017variational,tran2017hierarchical}, typically for latent variable models. Here, instead of performing maximum likelihood estimation  \citep{mescheder2017adversarial}, we  focus purely on VB  inference with approximate posteriors $q_{\b}(\theta\C X_0)$ when $\theta$ is assigned a prior. 

While \eqref{eq:beta_VB_oracle} provides a clear and elegant formulation, it encounters two training challenges: (1) classifier estimation when $\theta_0$ is unknown, and (2) instability under KL loss. To tackle these issues, we make two modifications: (1) we use the VB loss as a regularization term, and (2) we replace the discriminator function with the critic function. We elaborate the details below.

To address (1), note that we cannot train a classifier using the pair $\{\theta_0, X_0\}$ and $\{g_\beta(Z, X_0), X_0\}$ since $\theta_0$ is unknown. Instead, we train $d_{\phi} (X, \theta)$ on samples from  the joint $\pi(X, \theta)$ versus $q_\b(\theta\mid X)\pi(X)$, and then plug in $X_0$ as the condition variable. This gives rises to an adversarial game: the {\em Generator} $g_{\b}(Z,X)$ maximizes the ELBO conditioned on $X_0$, and the {\em Discriminator} $d_{\bphi}(X,\theta)$ distinguishes between true and generated joint distributions. However, such training schedule makes it challenging to maintain the balance between the generator network and the discriminator network, which is extremely intricate for GAN training. When the generator is learning locally and the critic function is learning globally, the balance is easily broken when the generator misbehaves away from the vicinity of $X_0$.

 We resolve this issue by incorporating the ELBO a regularization term in addition to the original adversarial loss, which encourages the generator improves locally near $X_0$. Basically the training alternates between updating the discriminator to maximizing the GAN loss and updating the generator to minimize the \{GAN loss +$\lambda\cdot$ELBO loss\}, where $\lambda$ controls the strength of the regularization.

 To address (2), i.e., KL-based training instability, we consider a different Wasserstein formulation motivated by  our theoretical results in Theorem 1.  Our first term associated with the total variation bound
\[
\mathcal A_1(\mF,{\mG})\equiv \E \inf_{\bomega:f_{\bomega}\in\mF}\left\| \log \frac{\pi(\cdot\C X)}{\pi_{{\wh\b_T}}(\cdot\C X )}-f_\omega(X, \cdot)\right\|_{\infty}
\]
suggests that the critic function $f_\omega(X, \cdot)$  also approximates the log-likelihood ratio $\log \frac{\pi(\cdot\C X)}{\pi_{{\wh\b_T}}(\cdot\C X )}$. This allows us to  rewrite the ELBO term in \eqref{eq:beta_VB_oracle} as
\begin{equation}\label{eq:ELBO_critic}
\wh{\mathcal{L}}(\beta) =\frac{1}{K}\sum_{\theta_i = g_{ \beta}(Z_i, X_0)}  f_{\hat\omega} (X_0, \theta_i).
\end{equation}
where $f_{\hat\omega}$ is a trained critic approximating the log-density ratio, based on samples from joint distributions $\pi(X, \theta)$ and $q_\b(\theta\mid X)\pi(X)$.

Putting it all together, we present the VB-regularized WGANs training procedure. We initialize the generator and critic networks at $\b^{(0)}$ and $w^{(0)}$ and generate the ABC reference table $\{(\theta_j,X_j)\}_{j=1}^T\iid\pi(\theta,X)$  with
$\{Z_j\}_{j=1}^T \iid \pi_Z(\cdot)$. Then, for the $t$-th iteration, update
\begin{align}
w^{(t)} &= \arg \max_{w: f_w\in \mF} \sum_{j=1}^T f_w(X_j, g_{\beta^{(t)}}(Z_j, X_j)) - \sum_{j=1}^T f_w(X_j, \theta_j) \nonumber \\
\beta^{(t)} & =\arg\min_{\beta: g_\beta \in \mG}  \underbrace{\sum_{j=1}^T f_{w^{(t)}}(X_j, g_\beta(Z_j, X_j))}_\text{GAN Loss} +  
\lambda \cdot \underbrace{\sum_{j=T+1}^{T+K} f_{w^{(t)}}(X_0, g_\beta(Z_j, X_0)) }_\text{approx VB loss in \eqref{eq:ELBO_critic}} \label{eq:VB_generator}
\end{align}
This formulation effectively balances global learning from simulated datasets with local adaptation to the observed data $X_0$, offering advantages over purely adversarial schemes such as \Cref{alg:WGAN}, which ignore $X_0$ altogether. Another perspective on \eqref{eq:VB_generator} is that it can also be viewed as a weighted loss between simulated datasets $\{X_j\}_{j=1}^T$ and the observed dataset $X_0$. We would thereby expect this version to work better than  \Cref{alg:WGAN} which does not use $X_0$ at all.

Indeed, on the toy simulated example in \Cref{ex:toy_RL} we can see that the VB variant produces tighter reconstructions relative to the B-GAN approach. The performance, however, is not uniformly better than Algorithm \ref{alg:WGAN-RL}. We provide a snapshot from another repetition in \Cref{fig:gauss_ex2_another} (\Cref{sec:supplement_VB}), where the spikiness of B-GAN-VB  (especially obvious when estimating $\theta_2$) may explain why the MMDs between B-GAN-VB and the true posteriors are larger than for B-GAN-2S in \Cref{fig:gauss_mmd}. 
The algorithm is described in \Cref{alg:WGAN-VB}. There are also other variants of adversarial VB using the classification-based KL loss or the Wasserstein loss \citet{mescheder2017adversarial,huszar2017variational}, but we have found that they are less stable to train and thus exclude them from the comparisons.

\begin{figure}[!t]
\centering
\begin{minipage}{.9\linewidth}
\begin{algorithm}[H]
	\caption{\em  GAN Variational Bayes (Wasserstein Version)}\label{alg:WGAN-VB}
\centering
	\spacingset{1.1}
	\resizebox{\linewidth}{!}{	
			\begin{tabular}{l l}
\hline
		\multicolumn{2}{c}{\bf INPUT \cellcolor[gray]{0.6} }\\	
				\hline
			\multicolumn{2}{c}{Prior $\pi(\theta)$, observed data $X_0$ and noise distribution $\pi_Z(z)$}\\
					\hline
			\multicolumn{2}{c}{\bf Training \cellcolor[gray]{0.6}}\\	
			\hline	
						\multicolumn{2}{c}{Initialize  network parameters {$\bomega^{(0)}=0$ and $\b^{(0)}=0$}}\\

			\multicolumn{2}{c}{\bf Pilot Run \cellcolor[gray]{0.9}}\\		
			\multicolumn{2}{c}{Apply  \Cref{alg:WGAN} with $\pi(\theta)$ to learn $\wh g_{pilot}(\cdot)$ }\\
			\multicolumn{2}{c}{\bf  \cellcolor[gray]{0.9}Reference Table}\\
				\qquad  Generate  pairs $\{(X_j,\theta_j)\}_{j=1}^T$ where $\theta_j=\wh g_{pilot}(Z_j,X_0)$ for $Z_j\sim\pi_Z$ and $X_j \sim P_{\theta_j}^{(n)}$.  &\\

			\multicolumn{2}{c}{\bf  \cellcolor[gray]{0.9}WGAN Training }\\
			\multicolumn{2}{l}{For $t=1,\dots, N$:}\\
			\quad {\bf Critic Update} ($N_\text{critic}$ steps): Same  as  in Critic Update of \Cref{alg:WGAN}\\
			\quad {\bf Generator Update} (single step) \\
			\qquad  \qquad  Generate noise $Z_j\sim\pi_Z(z)$ for  $j=1,\ldots, N$. \\
			\qquad  \qquad Update $\b^{(t)}$ by applying stochastic gradient descent on  \eqref{eq:VB_generator}. &\\
				
			\multicolumn{2}{c}{\bf  \cellcolor[gray]{0.9}Posterior Simulation }\\
			\qquad Simulate	 $\{Z_i\}_{i=1}^M\iid\pi_Z(z)$ and set $\wt \theta_i=  g_{\b^{(N)}}(Z_i, X_0)$.\\
			\qquad Estimate $\hat w_i$ using either \eqref{eq:kde_plugin} or \eqref{eq:weights}.\\
	\hline
		\multicolumn{2}{c}{\bf OUTPUT \cellcolor[gray]{0.6}}\\			
\hline
		\multicolumn{2}{c}{ Pairs of posterior samples and weights $(\wt \theta_1,\hat w_1), \ldots, (\wt \theta_M, \hat w_M)$ }\\
				\end{tabular}}
\end{algorithm}
\end{minipage}
\end{figure}

\section{Theory}\label{sec:theory}
The purpose of this section is to provide theoretical solidification for the   implicit posterior simulators in Algorithm \ref{alg:WGAN}, \ref{alg:WGAN-RL} and \ref{alg:WGAN-VB}. We will quantify the typical (squared) total variation (TV) distance between the actual posterior and its approximation and illustrate that with carefully chosen neural generators and discriminators, the expected total variation distance vanishes as $n\rightarrow\infty$. We will continue denoting $\Xn$ simply by $X$.

We define with $\nu=P_\theta^{(n)}\otimes \Pi$ the joint measure on $\mX\times\Theta$ with a density $\pi(X,\theta)=p_{\theta}^{(n)}(X)\pi(\theta)$. The goal is to approximate this measure with $\mu_g$ 
 defined semi-implicitly through a density function $\pi_g(X,\theta)=\pi(X)\pi_g(\theta\C X)$ where $\pi(X)=\int \pi(X,\theta)\d\theta$ is the marginal likelihood and where the samples from the density $\pi_g(\theta\C X)$ are obtained  by the {\em Generator} in Definition \ref{def:generator}. Thus, by keeping the marginal distribution the same, the distribution $\pi_g(\theta\C X)$ is ultimately approximating the conditional distribution $\pi(\theta\C X)$.
The quality of the approximation will be gauged under 
the integral probability metric\footnote{{The absolute value can be removed due to the Monge-Rubinstein dual (Villani, 2008).}
} (IPM)
\begin{equation}\label{eq:ipm}
d_{\mF}(\mu_g,\nu)=\sup_{f\in\mF}\left|E_{(X,\theta)\sim\mu_g}f(X,\theta)-E_{(X,\theta)\sim\nu}f(X,\theta)\right|,
\end{equation}
where $\mF$ is a class of evaluation metrics  (such as Lipschitz-1 functions that yield the Wasserstein-1).
The IPM metric \eqref{eq:ipm}, due to shared marginals of the two distributions, satisfies
\begin{align}
d_{\mF}(\mu_g,\nu)
\leq E_X d_{\mF}\big(\mu_g(X),\nu(X)\big),\label{eq:ineq_joint}
\end{align}
where $\mu_g(X)$ and $\nu(X)$ denote the {\em conditional} measures  with densities $\pi_g(\theta\C X)$ and $\pi(\theta\C X)$.
At the population level, the B-GAN (Algorithm \ref{alg:WGAN}) minimax game finds an equilibrium
$$
g^*=\min_{g\in\mG} d_{\mF}(\mu_g,\nu),
$$
where $\mG$ is a class of generating functions (that underlie the implicit distribution $\mu_g$). 

Typically, both $\mF$ and $\mG$ would be parameterized by neural networks with the hope that the discriminator networks can closely approximate the metric $d_\mF$ and that the generator networks can flexibly represent distributions. In practice, one would obtain a data-driven estimator based on the {\em empirical} distribution $\bar\nu_T$ of $(\theta_j,X_j)$ for $1\leq j \leq T$ and the {\em empirical} distribution $\bar\mu_g$ of  $(g(Z_j,X_j),X_j)$ where $Z_j\sim \pi_Z$ for $1\leq j\leq T$.  Assuming that $\mG=\{g_{\b}:\b\in\R^G\}$,  the B-GAN estimator can be written as 
\begin{equation}\label{eq:bhat}
\wh \b_T\in\arg\min_{\b:g_{\b}\in\mG}\max_{\bomega:f_{\bomega}\in \mF}\left| E_{\,\bar\mu_g} f_{\bomega}\big(X, g_{\b}(Z,X)\big)-E_{\,\bar\nu_T} f_{\bomega}(X, \theta)\right|.
\end{equation}
For brevity, we will often denote the generator density $\pi_{g_{\b}}(\cdot)$ (see Definition \ref{def:generator}) simply by $\pi_{\b}(\cdot)$ and similarly for $\mu_{g_{\b}}$.
The next Theorem provides an upper bound on the typical  total variation (TV) distance between true and the approximated posterior measures $\nu(X_0)$ and $\mu_{\wh\b_T}(X_0)$  with densities $\pi(\theta\C X_0)$ and $\pi_{\wh\b_T}(\theta\C X_0)$, respectively.    The total variation distance can be upper bounded by three terms: (1) the ability of the critic to tell the true model apart from the approximating model
\begin{equation}\label{eq:term1}
\mathcal A_1(\mF,{\mG})\equiv \E \inf_{\bomega:f_{\bomega}\in\mF}\left\| \log \frac{\pi(\cdot\C X)}{\pi_{{\wh\b_T}}(\cdot\C X )}-f_\omega(X, \cdot)\right\|_{\infty}
\end{equation}
(2) the ability of the generator to approximate the average true posterior
\begin{equation}\label{eq:term2}
\mathcal A_2(\mG)\equiv \inf_{\b:g_{\b}\in\mG} \left[E_{X}\left\|\log \frac{\pi_{\b}(\cdot\C X)}{\pi(\cdot\C X)}\right\|_\infty\right]^{1/2},
\end{equation}
and (3) the complexity of the (generating and) critic function classes   measured the pseudo-dimension $Pdim(\cdot)$ and defined in Definition 2 in \citet{bartlett2017spectrally}.  Intuitively, pseudo dimension is the maximum number of points you can assign arbitrary real values to and still be able to fit them exactly using functions from the hypothesis space. For example, if the hypothesis space is the set of all linear functions, its pseudo dimension is 2, since you can fit exactly 2 arbitrary points (with real-valued outputs), but no more than that, with a single linear function.

Note that the $L^\infty$ norm in both \eqref{eq:term1} and \eqref{eq:term2} are both taken with respect to $\theta$. We use $\cdot$ in place of $\theta$ to avoid ambiguity. We denote with $\mathcal H=\big\{h_{\bomega,\b}: h_{\bomega,\b}(Z,X)=f_{\bomega}
\big(g_{\b}(Z,X),X\big)\big\}$  a structured composition of networks $f_{\bomega}\in\mF$ and $g_{\b}\in\mG$.

\begin{thm}\label{thm:bound}
Let $\wh \b_T$ be as in  \eqref{eq:bhat} where $\mF=\{f: \|f\|_\infty\leq B\}$ for some $B>0$. Denote with $\E$ the {expectation   with respect to empirical measure on   $\{(\theta_j, X_j)\}_{j=1}^T$ and $\{Z_j\}_{j=1}^T$ in the reference table}.
 {Assume that   the prior satisfies 
\begin{equation}\label{eq:prior_conc}
\Pi[B_n(\theta_0;\epsilon)]\geq \e^{-C_2n\epsilon^2}\quad \text{for some  $C_2>0$ and $\epsilon >0$},
\end{equation} }
where the KL neighborhood $B_n(\theta_0;\epsilon)$ is defined as
\begin{equation}
B_n(\theta_0;\epsilon)=\{\theta\in\Theta: \KL(P_{\theta_0}^{(n)}\| P_{\theta}^{(n)})\leq n\epsilon^2, V_{2,0}(P_{\theta_0}^{(n)},P_{\theta}^{(n)})\leq n\epsilon^2\}.
\end{equation}
Then for $T\geq Pdim(\mF)\vee Pdim(\mathcal H)$ we have for any $C>0$
$$
P_{\theta_0}^{(n)}\E\,  d_{TV}^2\left(\nu(X_0),\mu_{\wh\b_T}(X_0)\right)\leq \mathcal C_n^T(\epsilon,C),
$$
where,  for some $\wt C>0$ and $Pmax\equiv Pdim(\mF)\vee Pdim(\mathcal H)$,
$$
C_n^T(\epsilon,C)=\frac{1}{C^2n\epsilon^2} +\frac{\e^{(1+C_2+C)n\epsilon^2}}{4} \left[2\mathcal A_1(\mF,{\mG})+
\frac{B\,\mathcal A_2(\mG)}{\sqrt 2}+4\wt C\,B \sqrt{\frac{\log T\times Pmax}{T}}\right].
$$
\end{thm}
\proof Section \ref{sec:proof_thm_bound}

\begin{remark}
\citet{zhou2022deep} provided theoretical results for the total variation distance between {\em joint} distributions using the Jensen-Shannon version of conditional GANs. Contrastingly, we provide frequentist-Bayesian results, quantifying the typical squared total variation distance  between the true and approximate posteriors after integrating over the data generating process. In addition, we use the Wasserstein GANs, building on oracle inequalities established in \citet{liang2021well}.
\end{remark}

Interestingly, while the GAN approach is trained {\em without} the knowledge of $X_0$, Theorem \ref{thm:bound} obtains a finite-sample (fixed $n$ and $d$)  upper bound for the typical TV distance {\em after plugging $X_0$} into the sampler. This frequentist result has required the prior concentration condition \eqref{eq:prior_conc} which could be avoided if we were willing to assume bounded likelihood ratios \citep{kaji2021mh} or uniform estimation error \citep{frazier2024statistical}. With bounded likelihood ratios, we would then also not have the first term and the exponential multiplier in the bound.  
 
From \eqref{eq:ineq_joint},  it can be seen that  Algorithm \ref{alg:WGAN}  targets a lower bound to the {\em average} Wasserstein distance between the posterior $\pi(\theta\C X)$ and $\pi_g(\theta\C X)$ after marginalizing over $\pi(X)$.  In other words, Algorithm \ref{alg:WGAN} {\em is not} necessarily targeting $\pi(\theta\C X_0)$. The 2-step enhancement in Algorithm \ref{alg:WGAN-RL} provides more  data draws in the ABC table that more closely resemble $X_0$.
 Theorem \ref{thm:bound} applies to Algorithm \ref{alg:WGAN-RL} as well with slight modifications.

\begin{corollary}(2step B-GAN)\label{cor:bound}
Assume that $\wh\b_T$  in  \eqref{eq:bhat}  is learned  under the proposal distribution $\wt\pi(\theta)$ and
denote with $\wt\E$ the expectation  of the reference table under $\wt\pi(\theta)$. 
Assume that the original prior $\pi(\theta)$ satisfies \eqref{eq:prior_conc}. Then the importance re-weighted posterior reconstruction from Algorithm \ref{alg:WGAN-RL} satisfies
the statement in Theorem \ref{thm:bound}  with $\E$ replaced by $\wt\E$ and with {\small$\wt{\mathcal A}_2(\mG)\equiv \inf\limits_{\b:g_{\b}\in\mG} \Big[\int_{\mX}\wt\pi(X)\left\|\log \frac{\pi_{\b}(\theta\C X)}{\pi(\theta\C X)}\right\|_\infty\d X \Big]^{1/2}$} and
{\small$\wt{\mathcal A}_1(\mF,{\mG})\equiv \wt\E \inf\limits_{\bomega:f_{\bomega}\in\mF}\left\| \frac{\pi(\theta)}{\wt\pi(\theta)}\log \frac{\pi(\theta\C X)}{\pi_{{\wh\b_T}}(\theta\C X )}-f_\omega(X, \theta)\right\|_{\infty}$}.
\end{corollary}

\proof Section \ref{sec:proof_cor_bound}

 The (nonasymptotic)  bounds for the {\em typical} TV distance   in Theorem \ref{thm:bound} and Corollary \ref{cor:bound} are not refined enough to fully appreciate the benefits of the 2-step enhancement.
In  Remark \ref{rem:RL} in Appendix (Section \ref{sec:motivation_RL}), we provide an intuitive explanation for why the 2-step refinement version works so well in practice on each particular realization of $X_0$ (not only on average over many realizations $X_0$).
We also provide a version of Theorem \ref{thm:bound} for adversarial variational Bayes (Algorithm \ref{alg:WGAN-VB}) in Theorem \ref{thm:vb_bound} (Section \ref{sec:thm_VB}).

Theorem \ref{thm:bound} provides intuition for how the quality of the generator and discriminator networks affects the total variation distance. It also suggests how large $T$ should be  (relative to $n$ and $d$) to assure that the average squared TV distance is arbitrarily small (vanishing with $n$). Recall that the training data size $T$ can be chosen by the user.  In Remark \ref{rem:cor} below, we will focus on specific deep learning architectures and provide a condition for choosing $T$ such that the upper bound in Theorem \ref{thm:bound} is $o(1)$ as $T\rightarrow\infty$ for a suitable choice of $\epsilon$ and $C$ (potentially depending on $n$) and for $n,d$ fixed. Instead of fixing $n$ and allowing $T$ to grow, in Corollary \ref{cor:dnn} we fix $T$ (as a function of $n$) and allow $n$ to grow and show that the upper bound is $o(1)$ as $n\rightarrow\infty$.  We use the ReLU activation function $\sigma_{ReLU}(x)=\max\{x,0\}$ for both the critic and generator networks, which have good approximation properties \citep{schmidt2020nonparametric}.

\begin{definition}
We denote with $\mathcal F_L^B(S,W)$ a class of feed-forward ReLU neural networks $f$ with depth $L$ (i.e. the number of hidden layers plus one), width $W$ and size $S$ (total number of parameters in the network) such that $\|f\|_\infty\leq B$. The width is defined as  $W=\max\{w_0,\dots,w_L\}$ where $w_l$ is the width of the $l^{th}$ layer  with $w_0$   the input data dimension and $w_L$ the output dimension. With $\mathcal G_L^B(S,W)$, we denote the leaky ReLU neural networks with the same meaning of parameters.
\end{definition}

The following Corollary  warrants optimism when using neural networks for the generator and the discriminator. We formulate the Corollary in context of  Algorithm \ref{alg:WGAN} and note that a similar conclusion holds for Algorithm \ref{alg:WGAN-RL} as well. This Corollary shows that there exist deep learning architectures such that for large enough $T$ (depending on $n$), the typical squared TV distance vanishes as $n\rightarrow \infty$.

\begin{corollary}\label{cor:dnn}
Assume that the joint distribution $\pi(\theta,X)$ is realizable in the sense that there exists $g_{\b_0}\in\mathcal G_{L_0}^{B_0}(S_0,W_0)$ such that $\pi(\theta,X)=
\pi_{g_{\b_0}}(\theta,X)$. Assume that $\mathcal G_{L^*}^{B^*}(S^*,W^*)\subseteq G_{L_0}^{B_0}(S_0,W_0)$ is a class of leaky ReLU generative networks indexed by $\b$ where $\|\b_0\|_\infty\vee\|\b\|_\infty\leq b$ for some $b>0$. 
Assume that $\mathcal F=\mathcal F_L^B(S,W)$ are ReLU discriminator networks and $\pi_Z$ is uniform on $[0,1]^d$.  Assume  the prior concentration \eqref{eq:prior_conc}  is satisfied with $\epsilon_n>0$ such that  $\epsilon_n=O(1/\sqrt{n})$. For each   arbitrarily slowly increasing sequence $C_n$,  there exists 
$L, S,W>0$  and training data size $T$ (depending on $n$) such that 
we have
$
P_{\theta_0}^{(n)} {\E}d_{TV}^2\left(\nu(X_0),\mu_{\wh\b_{T}}(X_0)\right)=o(1)\,\,\text{as $n\rightarrow\infty$ for $d$ fixed}.
$
\end{corollary}

\proof Section \ref{sec:proof_cor_dnn} in the Appendix.

\begin{remark}\label{rem:cor}
Using the same architecture as in Corollary 6, given $n$ and $d$ we can find the smallest $\epsilon$ such that the prior concentration is satisfied. Choosing $1/C=o(1)$ as $T\rightarrow\infty$  such that $\e^C\sqrt{(\log T\times Pmax)/T}=o(1)$ will yield an upper bound $C_n^T(\epsilon,C)$ that is $o(1)$ as $T\rightarrow \infty$ for fixed $n$ and $d$. 
\end{remark}

\section{Performance Evaluation}\label{sec:performance}
This section demonstrates very promising performance of our B-GAN approaches in  \Cref{alg:WGAN} (B-GAN), \Cref{alg:WGAN-RL} (B-GAN-2S) and  \Cref{alg:WGAN-VB} (B-GAN-VB)  and on simulated examples relative to other Bayesian likelihood-free methods 
(plain ABC using  summary statistics (SS);   2-Wasserstein distance ABC by \citet{bernton2019approximate};  Sequential Neural Likelihood (SNL) \citep{papamakarios2019sequential} with default settings). 
SNL \citep{papamakarios2019sequential} employs autoregressive flow to estimate the likelihood function from simulated ABC reference table.
The implementation details of our methods and the counterparts are described in \Cref{sec:implementation_lv} for the Lotka-Volterra model and \Cref{sec:implementation_bnb} for the Boom-and-Bust model.

 We evaluate the performance of different approximated posteriors in terms of three metrics: bias $\E\abs{\hat\theta_i -\theta_i}$, width of the 95\% credible intervals and its coverage and the Posterior Mass Concentration (PMC) in a small $\delta$ neighborhood of the true values $\btheta_0$. PMC is defined coordinate-wise as $\P\big(\hat \theta_j \in (\btheta_{0, j}-\delta_j, \btheta_{0,j}+\delta_j )\big)$ for every $j=1, \ldots, d$. PMC incorporates both bias and variance of the posterior approximation, and it is also helpful in understanding how much uncertainty has reduced compared to the prior. Since the examples we consider here all use uniform priors, we choose $\delta_j$ to be 5\% of the width of those uniform priors to visualize how much uncertainty has reduced after observing $X_0$.

\subsection{Lotka-Volterra Model}\label{sec:lv_main}

The Lotka-Volterra (LV) predator-prey model \citep{wilkinson2018stochastic} is one of the classical likelihood-free examples and describes population evolutions in ecosystems where predators interact with prey.  The state of the population is prescribed deterministically via a system of  ordinary differential equations (ODEs). Inference for such models is challenging because the transition density is intractable. However, simulation from the model is possible, which makes it a natural candidate for simulator-based inference methods. 

The model monitors population sizes of  predators $x_t$ and  prey $y_t$ over time $t$. The changes in states are determined by  four parameters $\btheta=(\theta_1, \ldots, \theta_4)'$ controlling: (1) the rate $r_1^t=\theta_1 x_t y_t$ of a predator being born; (2) the rate $r_2^t=\theta_2 x_t$ of a predator dying; (3) the rate $r_3^t=\theta_3 y_t$ of a prey being born; (4) the rate $r_4^t=\theta_4 x_t y_t$ of a prey dying.  Given the initial population sizes  $(x_0, y_0)$ at time $t=0$, the dynamics can be simulated using the Gillespie algorithm \citep{gillespie1977exact}, {which is a stochastic discrete-time Markov chain model.} The algorithm samples times to an event from an exponential distribution (with a rate $\sum_{j=1}^4 r_j^t$) and picks one of the four reactions with probabilities proportional to their individual rates $r_j^t$.  
We use the same  setup as \citet{kaji2021mh} where each simulation is started at $x_0=50$ and $y_0=100$ and state observations   are recorded every  0.1 time units for a period of 20 time units, resulting in a series of $201$ observations each. 

\begin{figure}[!t]
\centering
\includegraphics[width=0.95\textwidth, height=0.5\textwidth]{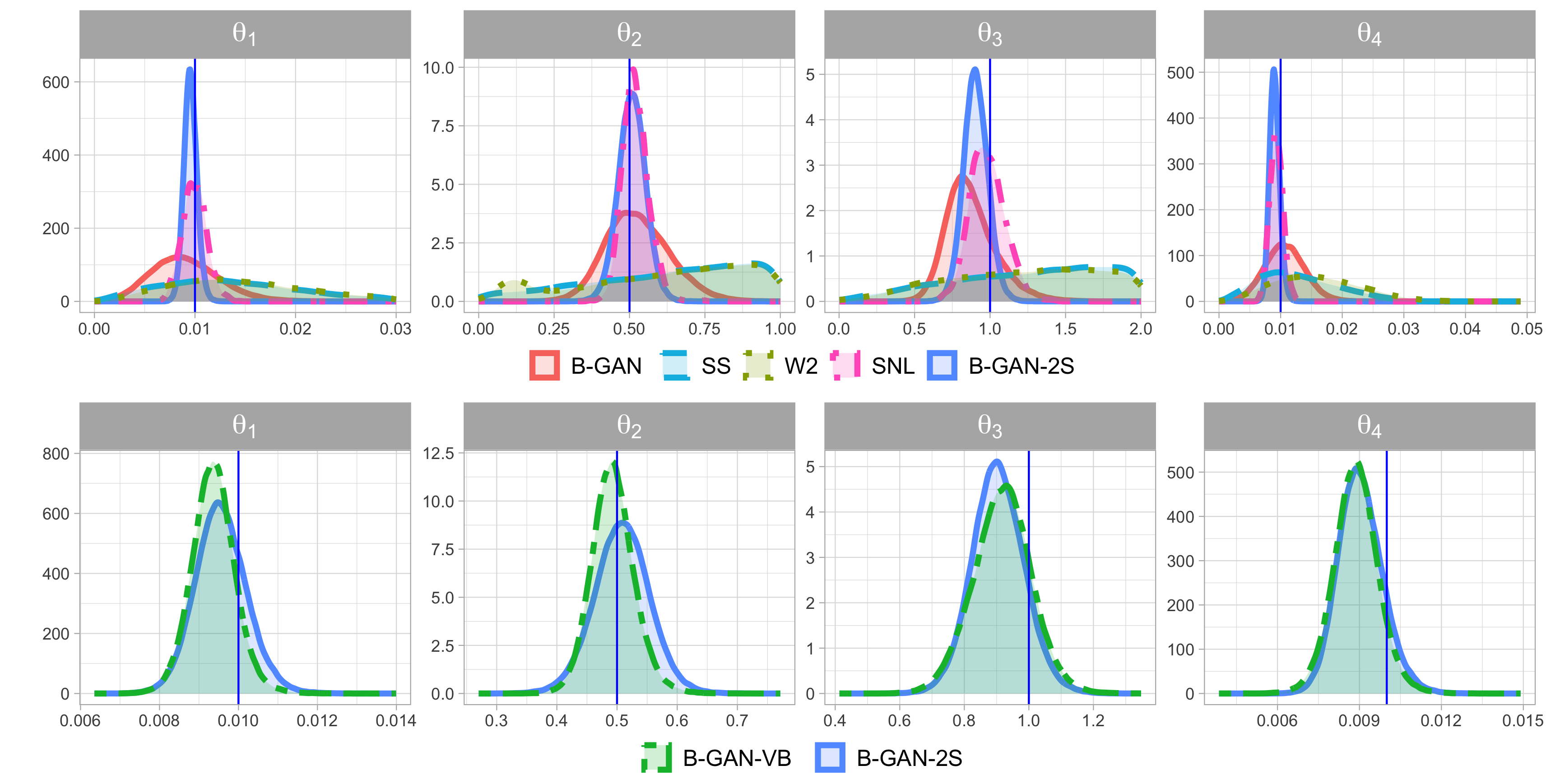}
\caption{\small Approximate posterior densities under the Lotka-Volterra Model. The true parameter vector (marked by vertical lines) is $\btheta_0=(0.01, 0.5, 1, 0.01)'$.}\label{fig:lv_post}
\end{figure}

The real data $X_0$  are generated with true values $\btheta_0=(0.01, 0.5, 1, 0.01)'$.  The data vector $X_0$ is stretched into one $(201\times 2\times n^*)$ vector, where $n^*$ is number of i.i.d. copies of time-series observations. The advantage of our approach is that it can be used even for $n^*=1$ when other methods (such as   \citet{kaji2021mh}) cannot. We focus on the  $n^*=1$ case here.
We  use an informative prior $\theta \in U(\Xi)$ with a restricted domain $\Xi =[0, 0.1] \times [0, 1] \times [0,2] \times [0, 0.1]$ to make it easier for classical ABC methods (see  \citet{kaji2021mh}) and to make the GAN training more efficient. Previous analyses \citep{papamakarios2016fast}  suggested summary statistics as the mean, log-variance, autocorrelation (at lag 1 and 2) of each series as well as their correlation. \citet{papamakarios2019sequential} also built their sequential neural network on top of this set of summary statistics. 
 We also build our model on the summary statistics, since it performs better than using the time series empirically.
  This example is quite challenging due to the spikiness of the likelihood in very narrow areas of the parameter space (as explained in \citet{kaji2021mh}).

A typical snapshot (for one particular data realization) of the approximated posteriors is given in  \Cref{fig:lv_post} and the summary statistics averaged over $10$ repetitions are reported in \Cref{tab:lv_post}.  Since we do not have access to the true posterior, we look at the width of the $95\%$ credible interval, its coverage (proportion of the 10 replications such that the true value is inside the credible interval),   bias of the posterior mean and the posterior mass concentration. 
Again, we observe that B-GAN-2S and B-GAN-VB outperforms B-GAN with  smaller biases, tighter variances and higher probability mass concentrated near $\btheta_0$. In \Cref{fig:lv_post}, B-GAN-VB appears to have smaller bias than B-GAN-2S when estimating all parameters and also consistently outperforms SNL. The computation cost requirements are compared in Section \ref{sec:compute_cost}.

\begin{table}[!t]
\centering
\scalebox{0.85}{
\begin{tabular}{*{2}{l} | *{6}{c}}
\toprule
 & & B-GAN & B-GAN-RL & B-GAN-VB & SNL & SS & W2 \\
\midrule
  \rowcolor{gray!50}  & Bias  & 0.30 & \bf{0.08} & \bf{0.08} & 0.11 & 0.96 & 1.10 \\
 \rowcolor{gray!50}  & CI   & 1.34 & 0.28 (0.9) & \bf{0.25 (0.9)} & 0.44 & 3.80 & 4.02 (0.9) \\
 \rowcolor{gray!50}  \multirow{-3}{*}{$\theta_1=0.01$}& PMC  & 0.82 & \bf{1.00} & \bf{1.00} & \bf{1.00} & 0.40 & 0.38 \smallskip \\
 \multirow{3}{*}{$\theta_2=0.5$} &Bias & 0.91 & 0.46 & \bf{0.43} & 0.45 & 2.49 & 2.42 \\
&CI   & 0.40 & 0.18 & \bf{0.14} & 0.17 & 0.91 & 0.84 \\
&PMC  & 0.33 & 0.62 & \bf{0.64} & \bf{0.64} & 0.10 & 0.10\smallskip \\
\rowcolor{gray!50} & Bias & 0.22 & \bf{0.12} & 0.13 & 0.13 & 0.49 & 0.47 \\
\rowcolor{gray!50}  & CI width   & 0.77 & \bf{0.41} & \bf{0.41} & 0.48 & 1.76 & 1.73 \\
\rowcolor{gray!50}  \multirow{-3}{*}{$\theta_3=1$}& PMC  & 0.27 & \bf{0.48} & \bf{0.48} & 0.47 & 0.11 & 0.11\smallskip \\
 \multirow{3}{*}{$\theta_4=0.01$} & Bias & 0.29 & \bf{0.12} & \bf{0.12} & 0.15 & 0.68 & 0.79 \\
& CI width  & 1.29 & 0.38 & \bf{0.35} & 0.52 & 2.72 & 2.82 \\
& PMC  & 0.83 & \bf{0.99} & \bf{0.99} & 0.98 & 0.51 & 0.42  \\
\bottomrule
\end{tabular}}
\caption{\footnotesize Summary statistics of the approximated posteriors under the Lotka-Volterra model (averaged over 10 repetitions). Bold fonts mark the best model of each column. The coverage of the $95\%$ credible intervals are 1 unless otherwise noted in the parentheses.}\label{tab:lv_post}
\end{table}

\subsection{Simple Recruitment, Boom and Bust}\label{sec:bnb_main}
 Our second  demonstration is on   the simple recruitment, boom and bust model  
  \citep{fasiolo2018extended}. 
The model is prescribed by a discrete stochastic process, characterizing the fluctuation of the population size of a certain
group over time. Given the population size $N_t$ and parameter $\btheta=(r, \kappa, \alpha, \beta)'$, the population size at the next timestep $N_{t+1}$ follows the following distribution 
\begin{equation*}
N_{t+1} \sim \left\{ \begin{array}{ll}
\text{Poisson} \big(N_t(1+r)\big)+\epsilon_t, &\text{ if } N_t\leq \kappa\\
\text{Binom} \big(N_t, \alpha\big)+\epsilon_t, &\text{ if } N_t > \kappa\\
\end{array}\right . ,
\end{equation*}
where $\epsilon_t \sim \text{Pois}(\beta)$ is a stochastic arrival process, with rate $\beta>0$.  The population grows stochastically at rate $r>0$, but it crashes if the carrying capacity $\kappa$ is exceeded. The survival probability $\alpha\in (0,1)$ determines the severity of the crash. Over time the population fluctuates between high and low population sizes for several cycles. 

This model has been shown to be extra challenging for both synthetic likelihood (SL) methods and ABC methods in \citet{fasiolo2018extended}. The distribution of the statistics is far from normal which breaks the normality assumption of SL. In addition, the authors show that ABC methods require exceedingly low tolerances and low acceptance rates to achieve satisfying accuracy.

We first run the simulation study using the  setup in \citet{an2020robust}. The real data $X_0$ is generated using parameters $r=0.4, \kappa=50, \alpha=0.09$ and $\beta=0.05$, and the prior distribution is uniform on $[0,1] \times [10, 80]\times [0,1]\times[0,1]$. The observed data has 250 time-steps, with 50 burn-in steps to remove the transient phase of the process.

Previous analyses of the model suggested various summary statistics, including the mean, variance, skewness, kurtosis of the data, lag 1 differences, and lag 1 ratios \citep{an2020robust}. We use them in SS and SNL methods. We have explored three types of input: the time series itself, the time series in conjunction with their summary statistics, and the summary statistics only.  We find that the network built on the summary statistics appears to perform the best, thus we only include the results from this network here.

\begin{figure}[!ht]
\centering
\includegraphics[width=0.9\textwidth]{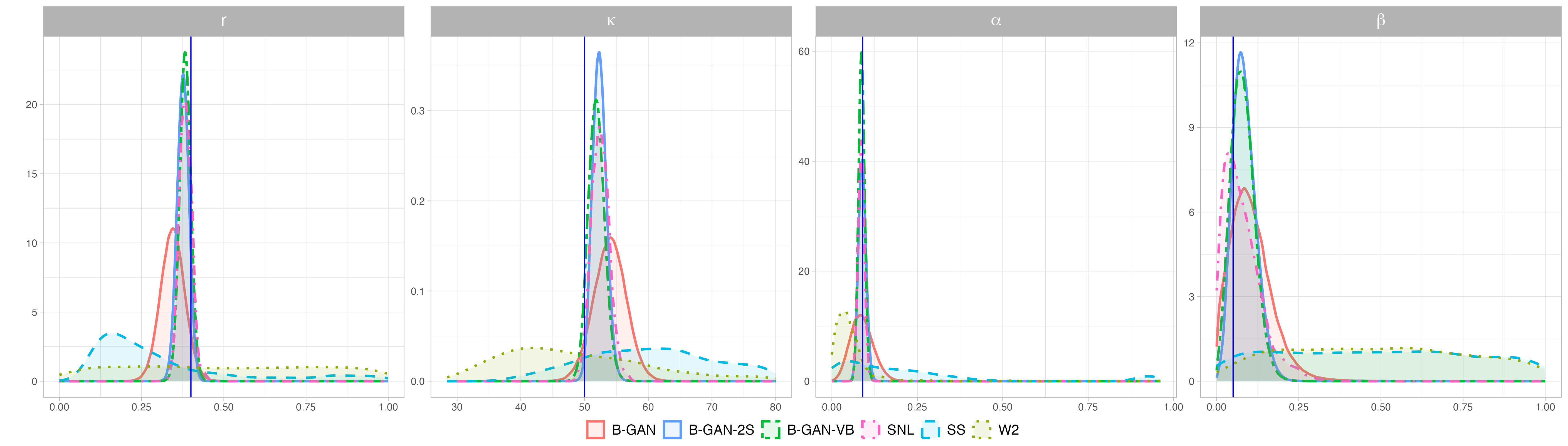}
\caption{\small Approximate posterior densities under the Boom-and-Bust Model. The true parameter is $\btheta_0=(0.4, 50, 0.09, 0.05)'$.  }\label{fig:bnb_post}
\end{figure}

We include SS, W2, SNL as competitors for comparisons. One snapshot of the approximate posterior densities is provided in \Cref{fig:bnb_post}. We report the performance summary averaged over 10 repetitions  in \Cref{tab:bnb_post}. ABC methods struggle to identify the parameters and provide very flat posteriors. The vanilla B-GAN is able to identify the correct location of parameters but with rather wide credible intervals. 
We observe great improvements after applying either the 2-step refinements or the VB refinements. For this example, B-GAN-VB performs the best in estimating $r$ and $\beta$, while B-GAN-RL approximates $\kappa$ better and SNL estimates $\alpha$ slighly better than B-GAN-VB. The BnB model is more challenging to learn, compared with the LV model. Potentially the performance of two refinements can be further improved if we add more sequential refinement steps with fewer training epoches to gradually modify the proposal distribution.

\begin{table}[!ht]
\centering
\scalebox{0.85}{
\begin{tabular}{*{2}{l} | *{6}{c}}
\toprule
 & & B-GAN & B-GAN-RL & B-GAN-VB & SNL & SS & W2 \\
\midrule
 \rowcolor{gray!50} 
& Bias ($\times 10^{-1}$) & 0.44 & 0.26 & \bf{0.24} & \bf{0.24} & 2.16 & 2.59 \\
\rowcolor{gray!50}  & CI width  ($\times 10^{-1}$) & 1.65 & 0.81 (0.9) & \bf{0.76 (0.9)} & 0.93 & 8.26 & 9.49 \\
\rowcolor{gray!50} \multirow{-3}{*}{$r=0.4$}  &PMC & 0.63 & 0.88 & \bf{0.91} & 0.90 & 0.08 & 0.09\smallskip \\
\multirow{3}{*}{$\kappa=50$} & Bias & 3.03 & \bf{1.42} & 1.56 & 1.52 & 10.60 & 10.16 \\
& CI width & 11.00 & \bf{4.33 (0.9)} & 5.09 & 5.37 & 37.17 & 43.20 \\
& PMC & 0.66 & \bf{0.96} & 0.94 & 0.94 & 0.22 & 0.19  \smallskip\\
\rowcolor{gray!50} & Bias ($\times 10^{-2}$) & 2.92 & 1.29 & 1.05 & \bf{1.01} & 15.08 & 5.46 \\
\rowcolor{gray!50} & CI width ($\times 10^{-1}$) & 1.37 & 0.50 & 0.39 & \bf{0.38} & 9.18 & 2.77 \\
\rowcolor{gray!50}  \multirow{-3}{*}{$\alpha=0.1$}& PMC & 0.83 & \bf{1.00} & \bf{1.00} & \bf{1.00} & 0.31 & 0.55 \smallskip\\
\multirow{3}{*}{$\beta=0.05$} & Bias ($\times 10^{-1}$) & 1.20 &\bf{0.96} & 1.01 & 1.28 & 4.41 & 3.92 \\
&CI width & 0.39 (0.9) & 0.24 (0.7) & \bf{0.23 (0.8)} & 0.39 (0.9) & 0.95 & 0.86 (0.5) \\
& PMC & 0.44 & \bf{0.54} & \bf{0.54} & 0.52 & 0.11 & 0.12 \\
\bottomrule
\end{tabular}}
\caption{\footnotesize Summary statistics of the approximated posteriors under the Boom-and-Bust model (averaged over 10 repetitions). Bold fonts mark the best model of each column. The coverage of the $95\%$ credible intervals are 1 unless otherwise noted in the parentheses.}\label{tab:bnb_post}
\end{table}

\section{Empirical Analysis: the Susceptible-Infected-Recovered (SIR) epidemic model with application to the common cold data}\label{sec:real_data}

In this section, we illustrate our approach on the problem of estimating the three-parameter Susceptible-Infected-Recovered (SIR) epidemic model and apply our methods to the 21-day common-cold outbreak data in Tristan da Cunha island  from October 1967 (Hammond \& Tyrrell (1971); Shibli et al. (1971)).  We provide the data in \Cref{tab:common_cold}, which tracks the number of infected and recovered individuals. 

\begin{table}[!ht]
	\centering
	 \resizebox{0.8\textwidth}{!}{
	\begin{tabular}{l |*{21}{c}}
		\toprule
		Day&1 &2 &3 &4 &5 &6 &7 &8 &9 &10 &11 &12 &13 &14 &15 &16 &17 &18 &19 &20 &21 \\
		\midrule
		Infected &  1 &1& 3& 7& 6& 10& 13 &13 &14 &14 &17 &10 &6 &6 &4& 3& 1& 1& 1& 1& 0 \\
		Recovered & 0 &0 &0 &0 & 5 &7 &8 &13 &13 &16 &16 &24 &30 &31 &33 &34 &36 &36 &36 &36 &37 \\
		\bottomrule
	\end{tabular}
	}
	\caption{Common cold outbreak data in Tristan da Cunha island (1967).}\label{tab:common_cold}
\end{table}

The SIR model categorizes hosts into three statuses at time $t$. Individuals are considered susceptible (S), if they are able to be infected by the pathogen, infected (I) if currently infected with the pathogen or recovered (R) if they have successfully cleared the pathogen. SIR models and their variants, in both deterministic and stochastic forms, are among the most fundamental epidemiological models and have found use describing diseases as diverse as influenza, herpes and malaria. We adopt the Model (3.12) in \citet{toni2009approximate}, which has the highest probability, to describe the dynamics. This model introduces a latent state where individuals are infected but not yet infectious, to account for the delay between infection and the ability to infect others.

The deterministic SIR model without demographics can be written mathematically as
\begin{align*}
\frac{\d S(t)}{\d t}  = -\beta I(t)S(t), \quad \frac{\d L(t)}{\d t}= \beta I(t)S(t) - \delta L(t), \quad
\frac{\d I(t)}{\d t} = \delta L(t) - \gamma I(t), \quad \frac{\d R(t)}{\d t} = \gamma I(t),
\end{align*}
where $S(t), L(t), I(t), R(t)$ are the numbers of susceptible, latent, infected and recovered individuals in the population at time $t$ (days), and $\delta$ is the rate at which latent individuals become infectious. The parameter $\beta$ is the transmission rate, and $\gamma$ is the recovery rate. 

Similar to the Lotka-Volterra model,  we use stochastic discrete-time Markov chain to simulate the dynamics. The model is governed by four parameters $\theta=(\beta, \gamma,\delta, S(0))'$, where $S(0)$ is the number of susceptible individuals at week $0$. Since the data in \Cref{tab:common_cold} only collected the number of infected and recovered individuals, the number of susceptible individuals is not directly observed and thus $S(0)$ needs to be estimated. 

Following the setting in \citet{toni2009approximate}, we place a prior on the parameters $\theta$ as $S(0)\sim \text{U}[37,100], \beta\sim \text{U}[0,3], \gamma\sim \text{U}[0,3],\delta\sim \text{Unif}[0,5]$. We include the plot of approximated posterior in  \Cref{fig:sir_post}. For $\beta$ and $S(0)$, the posterior seems to be unimodal, similar to the findings in \citet{toni2009approximate}, and B-GAN-VB provides the tightest credible intervals. The posterior of $\gamma$ and $\delta$ appears to be multimodal, which makes the inference more challenging and the posterior we obtain here is more dispersed than that in \citet{toni2009approximate}.

\begin{figure}[!ht]
	\centering
	\includegraphics[width=0.9\textwidth]{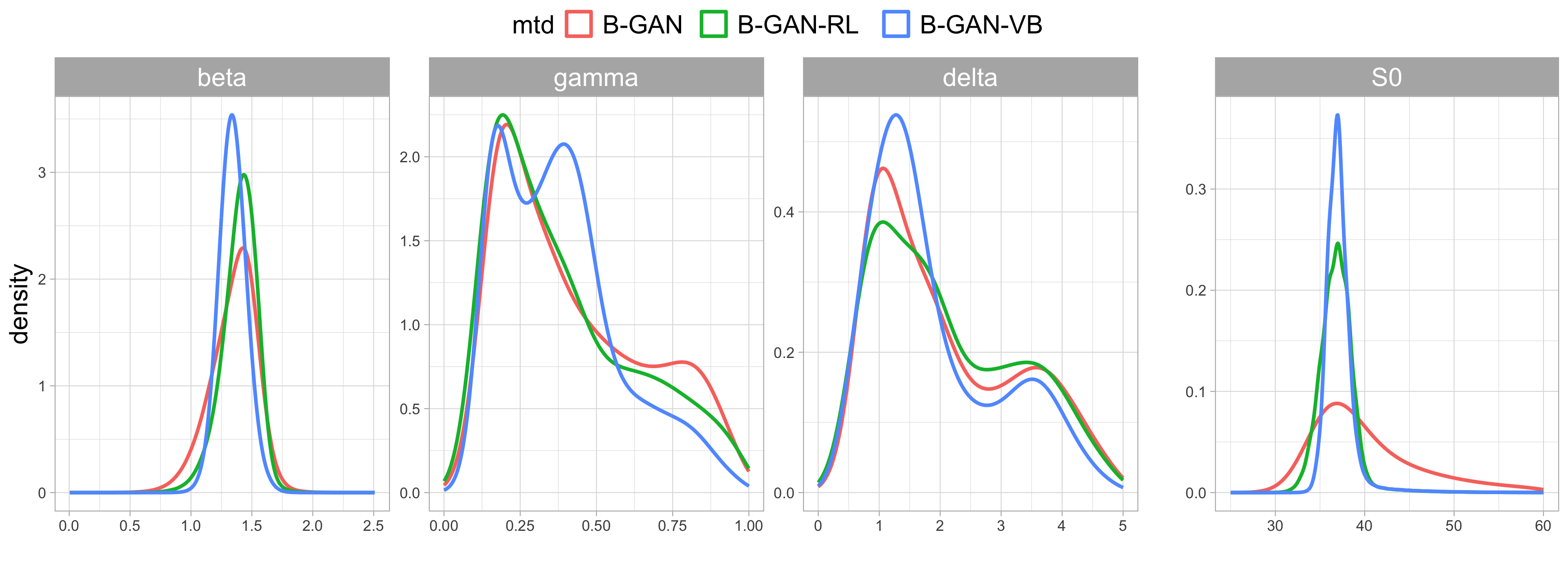}
	\caption{Approximated posterior densities under the SIR model.}\label{fig:sir_post}
\end{figure}

To further investigate the quality of the approximated posteriors, we examine the predictive performance of the fitted models. We simulate 100 datasets from the posterior predictive distribution and compare them with the observed data $X_0$. The results are provided in \Cref{fig:sir_pred}. We observe that for all methods, the simulated datasets cover the observed data, with B-GAN-VB providing the best fit.

\begin{figure}[!ht]
	\centering
	\includegraphics[width=0.9\textwidth]{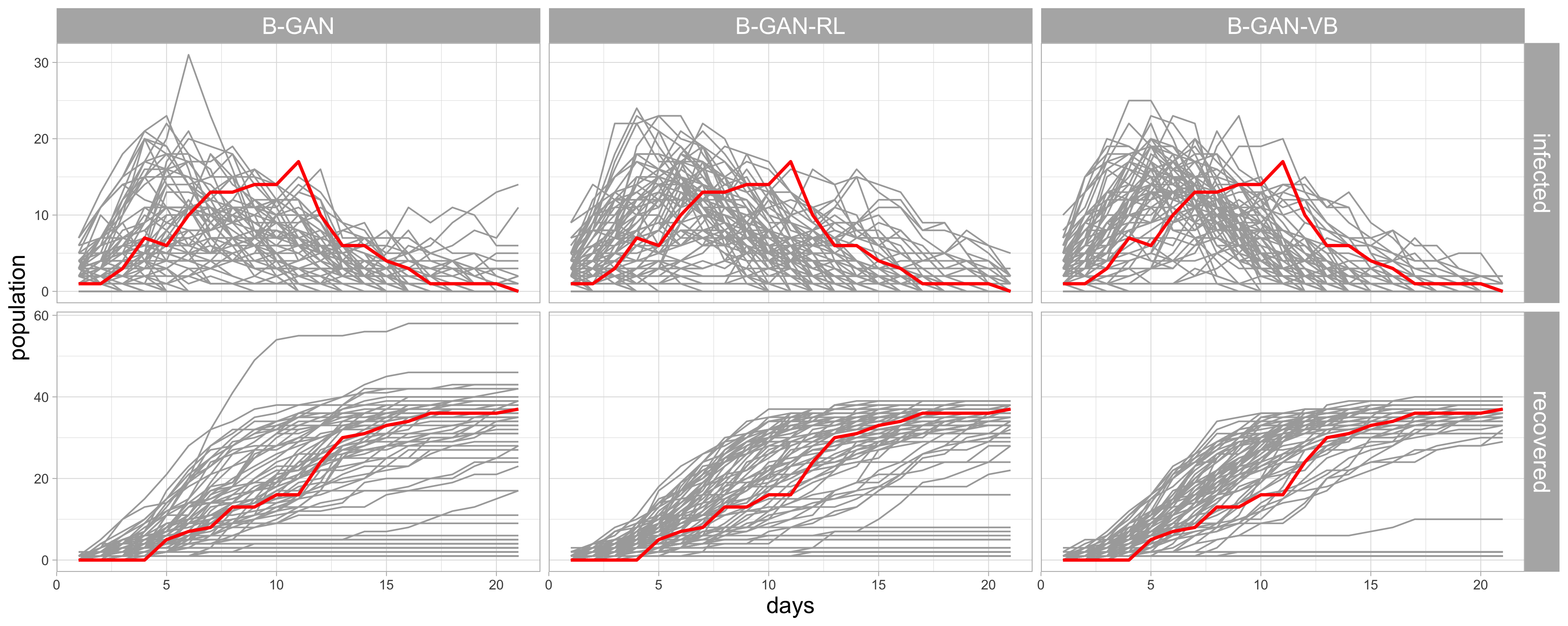}
	\caption{Simulated datasets under the posterior distributions. The red curves are the observed data. Each grey line represents one simulated time series of the infected/recovered population.}\label{fig:sir_pred}
\end{figure}

\section{Discussion}\label{sec:discussion}
This paper proposes strategies for Bayesian simulation using generative networks. We have formalized several schemes for implicit posterior simulation using  GAN conditional density regression estimators  as well as implicit variational Bayes. The common denominator behind our techniques is (joint) contrastive adversarial learning  \citep{tran2017hierarchical, huszar2017variational}. We have provided firm theoretical support in the form of bounds for the typical total variation distance between the posterior and its approximation. We have highlighted the potential of our adversarial samplers on several simulated examples with very encouraging findings. We hope that our paper will  embolden  practitioners to implement neural network posterior samplers in difficult situations when  likelihood (and prior) are implicit. 

 One advantage of using conditional GANs over other conditional density estimators, such as the autoregressive flow in \citet{papamakarios2019sequential}, lies in their expressive flexibility and ability to incorporate statistical structures. While the two methods share a similar purpose, they are fundamentally different. SNL provides a closed-form representation for the approximated density function, while the generator $g(\cdot, X^{(n)})$ in WGAN is approximating a pushforward mapping from $\pi_Z$ to the conditional distribution $\pi(\theta\mid X^{(n)})$ and no closed-form of the density function is available. Due to the difference, SNL is more likely to suffer from the curse of dimensionality due to the invertibility constraint \citep{papamakarios2021normalizing}. In addition, the invertibility constraint also imposes restrictions on choices of neural network architecture. In constrast, the generator in our WGAN framework can be designed to accommodate various structures that exploit the statistical properties of the data. For instance, in the case of i.i.d. data  as discussed in \Cref{sec:iid} or  partially exchanegable data as in \citet{luciano2025permutations}, the generator can be designed in a parameter-efficient manner that exploits the exchangeability of the data, thereby avoiding exploding network complexity as the number of observations increases. Moreover, when the posterior is approximately Gaussian, the generator can be constructed as two neural networks that separately parameterize the posterior mean and variance.  More recently, \citep{baptista2024conditional} consider the monotonicity constraint for the generator, which could offer extra inferential benefits with its connection to the optimal transport map.

Another interesting direction is to extend our approach to the case of mis-specified models, where the true data is not consistent with the model. Understanding the behavior of the posterior in such cases is crucial for practical applications. We leave the exploration of this topic for future work.
Some possible directions include incorporating the idea of Bayesian predictive check \citep{guttman1967use,rubin1984bayesianly} to diagnose the model mis-specification and  estimating the conditional distribution given  robust summary statistics.

\section*{Acknowledgments}
The authors gratefully acknowledge support from the James S. Kemper Research Fund at the Booth School of Business and the National Science Foundation (DMS: 1944740, DMS: 2515542). This work used the DeltaAI system at the National Center for Supercomputing Applications through allocation MTH250021 from the Advanced Cyberinfrastructure Coordination Ecosystem: Services \& Support (ACCESS) program, which is supported by National Science Foundation grants \#2138259, \#2138286, \#2138307, \#2137603, and \#2138296. We also thank the editor and the anonymous reviewers for their constructive comments and suggestions, which have significantly improved the quality of this paper.

\bibliography{cgan}

\newpage

\appendix

\section{Theory for Jensen-Shannon Conditional GANs}\label{sec:JS_detail}

\begin{definition}(Discriminator)
We define a deterministic {\em discriminative model}  as a mapping $d:(\mathcal X\times \Theta)\rightarrow (0,1)$   which predicts whether the data pair $(X,\theta)$ came from $\pi(X,\theta)$ (label $1$) or from $\pi_{g}(X,\theta)$ (label $0$).
 \end{definition}

\begin{lemma}{\citep[Lemma 2.1]{zhou2022deep}}\label{lem:noise_outsourcing} Let $(X, \theta)$ be a random pair taking values in $\mX\times \Theta$ with a joint distribution $\pi(X,\theta)$. Then,  for any given ${d_z}\geq 1$, there exists a random vector $Z\sim \pi_Z = N(0, I_{d_z})$  and a Borel-measurable function $g^*: \R^{d}\times \mX \to \Theta$ such that $Z$ is independent of $X$ and 
$(X, \theta) = (X, g^*(Z, X))$ almost surely. 
\end{lemma}

\begin{lemma}\label{lemma:cGAN}
Consider a minimax game
$(g^*,d^*)=\arg\min\limits_{g\in\mathcal G}\max\limits_{d\in\mathcal D} D(g,d)$
prescribed by
\begin{equation}\label{eq:gan_obj}
D(g,d)=E_{(X,\theta)\sim\pi(X,\theta)}\log d(X,\theta)+E_{X\sim \pi(X),Z\sim\pi_Z}\log[1-d(X,g(Z,X)].
\end{equation}
Assume that $\mathcal G$ and $\mathcal D$ are universal approximators capable of representing any function  $g:(\mathcal Z\times \mathcal X)\rightarrow \Theta$ and $d:(\mathcal X\times \Theta)\rightarrow (0,1)$, respectively. Then, uniformly on $\mX$ and $\Theta$, the solution $(g^*,d^*)$   satisfies 
\begin{equation*}
\pi_{g^*}(\theta\C X)=\frac{\pi(X,\theta)}{\pi(X)}=\pi(\theta\C X)\quad\text{and}\quad d^*_g(X,\theta)=\frac{\pi (X,\theta)}{\pi(X,\theta)+\pi_g(X,\theta)} \quad\text{for any $g\in\mathcal G$}.
\end{equation*}
\end{lemma}
\begin{proof}
The expression for $d^*_g(X,\theta)$ is an immediate consequence of Proposition 1 in  \citet{goodfellow2014generative}.
Plugging-in this expression into \eqref{eq:gan_obj}, we find that
$$
g^*=\argmin_{g\in\mathcal G} \left( E_{(X,\theta)\sim\pi(X,\theta)}\log d^*_g(X,\theta)+E_{X\sim \pi(X),z\sim\pi_Z}\log[1-d^*_g(X,g(Z,X))]\right),
$$
According to Theorem 1 of  \citet{goodfellow2014generative}, the minimum is achieved if and only if $\pi_{g^*}(X,\theta)={\pi(X,\theta)}=\pi(\theta\C X)\pi(X)$. The fact that  $\pi(X,\theta)$ and $\pi_{g^*}(X,\theta)$ have the same marginal $\pi(X)$ implies 
the expression for $\pi_{g^*}(\theta\C X)$.
\end{proof}

\section{Proofs from \Cref{sec:theory}}\label{sec:proof}

\subsection{Proof of \Cref{thm:bound}}\label{sec:proof_thm_bound}

\begin{proof}
{We continue denoting $\Xn$ simply by $X$.} Recall the definition of the KL neighborhood 
\begin{equation}\label{eq:KL}
B_n(\theta_0;\epsilon)=\{\theta\in\Theta: \KL(P_{\theta_0}^{(n)}\| P_{\theta}^{(n)})\leq n\epsilon^2, V_{2,0}(P_{\theta_0}^{(n)},P_{\theta}^{(n)})\leq n\epsilon^2\},
\end{equation}
where
$\KL(P_{\theta_0}^{(n)}\| P_{\theta}^{(n)})=P_{\theta_0}^{(n)}\log[p_{\theta_0}^{(n)}/p_{\theta}^{(n)}]$ and
\begin{equation}\label{eq:variance}
V_{2,0}(P_{\theta_0}^{(n)},P_{\theta}^{(n)})=P_{\theta_0}^{(n)}\abs{\log[p_{\theta_0}^{(n)}/p_{\theta}^{(n)}]-\KL(P_{\theta_0}^{(n)}\| P_{\theta}^{(n)})}^2.
\end{equation}
We define an event, for some fixed $C>0$ and $\epsilon>0$,
$$
\mA_n(\epsilon)=\left\{X: \int_{B_n(\theta_0;\epsilon)}\frac{p_\theta^{(n)}(X)}{p_{\theta_0}^{(n)}(X)}\pi(\theta)\d\theta>\e^{-(1+C)n\varepsilon^2}\Pi[B_n(\theta_0;\epsilon)]\right\}.
$$
{We denote with $\E$ the expectation with respect to $\{(\theta_j,X_j)\}_{j=1}^T$ from the ABC reference table sampled from the joint $\pi(\theta, X)$. For simplicity of notation, we use $\E_{\wh\b_T}$ interchangeably with $\E$, since it is equivalently accounting for the randomness in $\wh\b_T$.} Using the fact that the total variation distance is bounded by $2$,  we have
\begin{align}
P_{\theta_0}^{(n)}{\E_{\hat\beta_T}} d_{TV}^2\left(\nu(X_0),\mu_{\wh\b_T}(X_0)\right)=&\int_{\mA_n(\epsilon)}p_{\theta_0}^{(n)}(X_0) {\E_{\hat\beta_T}}  d_{TV}^2\left(\nu(X_0),\mu_{\wh\b_T}(X_0)\right)\d X_0\label{eq:decompose}\\
&+4\, \P_{\theta_0}^{(n)}[\mA_n^c(\epsilon)].
\end{align}
According to Lemma 10 in \cite{ghosal2007convergence}, we have $\P_{\theta_0}^{(n)}[\mA_n^c(\epsilon)]\leq \frac{1}{C^2n\epsilon^2}$.  Denoting with $\pi(X)=\int_\theta p_\theta^{(n)}(X)\pi(\theta)\d\theta$ the marginal likelihood,  we can rewrite the term in \eqref{eq:decompose} as
$$
{\int_{\mX}\mathbb I_{\mA_n(\epsilon)}}(X)\pi(X)r(X) {\E_{\hat\beta_T}} d_{TV}^2\left(\nu(X),\mu_{\wh\b_T}(X)\right)\d X
$$
where 
$$
\frac{1}{r(X)}\equiv{\int_{\Theta}\frac{p_{\theta}^{(n)}(X)}{p_{\theta_0}^{(n)}(X)}\pi(\theta)\d\theta}\geq {\int_{  B_n(\theta_0;\epsilon)}\frac{p_{\theta}^{(n)}(X)}{p_{\theta_0}^{(n)}(X)}\pi(\theta)\d\theta}.
$$
On the event $\mA_n(\epsilon)$, we can thus write
$$
r(X)< \frac{\e^{(1+C)n\epsilon^2}}{\Pi[B_n(\theta_0;\epsilon)]}.
$$
Under the assumption \eqref{eq:prior_conc}, the term in \eqref{eq:decompose} can be  upper bounded by
\begin{equation}\label{eq:until}
\e^{(1+C+C_2)n\epsilon^2}\int_{\mX}\pi(X) {\E_{\hat\beta_T}} d_{TV}^2\left(\nu(X),\mu_{\wh\b_T}(X)\right)\d X \leq \frac{\e^{(1+C+C_2)n\epsilon^2}}{4} {\E_{\hat\beta_T}} [{\KL(\nu  \| \mu_{\wh\b_T})+\KL(\mu_{\wh\b_T}\| \nu)}].
\end{equation}
The inequality above stems from the Pinsker's inequality \citep[Theorem 4.8]{van2014probability}  and the fact that the joint measures $\nu$ and $\mu_{\wh\b_T}$ have the same marginal distribution $\pi(X)$. In particular, {using Fubini's theorem}, we can write
\begin{align*}
&\int_{\mX}\pi(X) {\E_{\hat\beta_T}} d_{TV}^2\left(\nu(X),\mu_{\wh\b_T}(X)\right)\d X\\
&\quad\leq \frac{1}{4}\int_{\mX}\pi(X) {\E_{\hat\beta_T}} [\KL(\nu(X) \| \mu_{\wh\b_T}(X))+\KL(\mu_{\wh\b_T}(X) \| \nu(X))]\d X\\
&\quad= \frac{1}{4} {\E_{\hat\beta_T}}  \int_{\mX}\pi(X)  \int_\Theta \log\frac{\pi(\theta\C X)}{\pi_{\wh\b_T}(\theta\C X)}\left[\pi(\theta\C X)-\pi_{\wh\b_T}(\theta\C X)\right]\d \theta \d X \\
&\quad= \frac{1}{4} {\E_{\hat\beta_T}}  \left[{\KL(\nu \| \mu_{\wh\b_T})+\KL(\mu_{\wh\b_T} \| \nu)}\right]\equiv \frac{1}{4} {\E_{\hat\beta_T}} d_{\KL}^S(\nu, \mu_{\wh\b_T}).
\end{align*}
The above inequality is essential for understanding how the average squared total variation distance between the posterior and its approximation (with the average taken with respect to the observed data generating process) can be related to the `symmetrized' KL divergence $d_{\KL}^S(\nu, \mu_{\wh\b_T})$ between the {\em joint} distribution and its approximation. We now continue to bound the symmetrized KL divergence.
For simplicity, we denote with $\wh\b$ the estimator $\wh\b_T$ in \eqref{eq:bhat}. We have the following decomposition, for any $\bomega$ such that $f_{\bomega}\in\mF$,
 \begin{align*}
d_{\KL}^S(\nu, \mu_{\wh\b})=& \int_{\mX}\pi(X)\int_\Theta f_\omega(\theta,X)[\pi(\theta\C X)-\pi_{\wh\b}(\theta\C X)]\d\theta \d X\\
&+\int_{\mX}\pi(X) \int_{\Theta}\left[ \log \frac{\pi(\theta\C X)}{\pi_{{\wh\b}}(\theta\C X )}-f_\omega(\theta,X)\right][\pi(\theta\C X)-\pi_{{\wh\b}}(\theta\C X)]\d\theta \d X\\
\leq& \,
d_\mF\left(\nu,\mu_{\wh\b}\right)+2\left\| \log \frac{\pi(\theta\C X)}{\pi_{{\wh\b}}(\theta\C X )}-f_\omega(\theta,X)\right\|_{\infty},
\end{align*}
where we have used the inequality $\int f g\leq \|f\|_\infty\|g\|_1$ and the fact that  $\pi(\theta\C X)$ and $\pi_{{\wh\b}}(\theta\C X)$ are both non-negative and integrate to one.
Then, choosing $f_{\bomega}\in\mF$ that minimizes the second term we obtain
$$
d_{\KL}^S(\nu, \mu_{\wh\b})\leq 2\, \mathcal A_1(\mF, \mG,\wh\b) + d_\mF\left(\nu,\mu_{\wh\b}\right),
$$
{where $\mathcal A_1(\mF,{\mG},\wh\b)$ is defined as
\[\mathcal A_1(\mF,{\mG},\wh\b)\equiv  \inf_{\bomega:f_{\bomega}\in\mF}\left\| \log \frac{\pi(\theta\C X)}{\pi_{{\wh\b_T}}(\theta\C X )}-f_\omega(X, \theta)\right\|_{\infty}\]
and $\mathcal A_1(\mF,\mG)=\E_{\wh\b}\mathcal  A_1(\mF,\mG,\wh\b)$ was defined in \eqref{eq:term1}.}

We now apply a mild modification of the oracle inequality in \citep[Lemma 12]{liang2021well}. As long as $\mF$ and $\mathcal H$ are symmetric\footnote{i.e. if $f\in\mF$ then also $-f\in\mF$}, then  for any $\b$ such that $g_{\b}\in\mG$ we have
\begin{equation}\label{eq:oracle}
d_\mF(\nu,\mu_{ \wh\b})\leq d_\mF(\mu_{\b},\nu)+2d_\mF(\bar \nu_T,\nu)+d_\mF(\bar\mu^{\b}_T,\mu_{\b})+d_{\mathcal H}(\bar \pi_T,\pi),
\end{equation}
where $\bar\nu_T$ and $\bar\mu^{\b}_T$ are the empirical counterparts of $\nu, \mu_{\b}$ based on $T$ iid samples (ABC reference table $\{(\theta_j,X_j)\}_{j=1}^T$ for $\nu$ and $\{(g_{\b}(Z_j,X_j),X_j)\}_{j=1}^T$ with $Z_j\iid \pi_Z$ for $\mu_{\b}$). In addition $\bar\pi_T$ is the empirical version for the distribution $\pi_Z$ for $Z$
and   $\mathcal H=\{h_{\bomega,\b}: h_{\bomega,\b}(Z,X)=f_{\bomega}(X,g_{\b}(Z,X))\}$. This oracle inequality can be proved analogously as \citep[Lemma 12]{liang2021well} the only difference being that due to the conditional structure of our GANs the function class $\mathcal H$ is not entirely a composition of networks $f_{\bomega}$ and $g_{\b}$ but has a certain nested structure. 
Similarly as in \citep{liang2021well} (proof of Theorem 13), we  can write for any $\b$ such that  $g_{\b}\in\mG$
\begin{align*}
d_\mF(\mu_{\b},\nu)&\leq B\times d_{TV}(\mu_{\b},\nu)\leq B\,\sqrt{\frac{1}{4}d_{\KL}^S(\mu_{\b},\nu)}\\
&\leq \frac{B}{\sqrt 2}\left[{\int_{\mX}\left\|\log \frac{\pi_{\b}(\theta\C X)}{\pi(\theta\C X)}\right\|_\infty  {\pi(X)} \d X}\right]^{1/2}.
\end{align*}
Choosing $\b$ that minimizes the expectation on the right side, we obtain
$d_\mF(\mu_{\b},\nu)\leq \frac{B}{\sqrt{2}}\mathcal A_2(\mG)$, where the term $A_2(\mG)$ was defined in \eqref{eq:term2}.
Denote with 
$$
R_T(\mF)=E_{\bm\varepsilon}\sup_{f\in\mathcal F}\frac{1}{T}\sum_{j=1}^T\varepsilon_j f(X_j,\theta_j)
$$ 
the Rademacher complexity with $\bm\varepsilon=(\varepsilon_1,\dots,\varepsilon_T)'$ iid Rademacher\footnote{taking values $\{-1,+1\}$ with probability $1/2$.} variables.
For the second term in \eqref{eq:oracle}, the symmetrization property (see e.g. Lemma 26 in \citet{liang2021well} or Theorem 3.17 in \citet{sen2018gentle})
yields  for $T\geq Pdim(\mathcal F)$
$$
\E\, d_\mF(\bar \nu_T,\nu)\leq 2\, \E\, R_T(\mF)\leq \wt C\times B\sqrt{\frac{Pdim(\mathcal F)\log T}{T}}
$$
for some $\wt C>0$, where  $Pdim(\mathcal F)$ is the pseudo-dimension of the function class $\mathcal F$ (Definition 2 in \citep{bartlett2017spectrally}). The second inequality follows, for example, from Lemma 29 in \citep{liang2021well}. The bounds on 
$\E\, d_{\mathcal H}(\bar \pi_T,\pi),$ and  $\E\, d_\mF(\bar\mu^{\beta}_T,\mu_\beta)$ in \eqref{eq:oracle} are analogous.
Putting the pieces together from the oracle inequality in \eqref{eq:oracle} we can upper-bound $P_{\theta_0}^{(n)}\E d_{TV}^2\left(\nu(X_0),\mu_{\wh\b_T}(X_0)\right)$ with
\begin{align*}
\frac{1}{C^2n\varepsilon^2} +\frac{\e^{{(1+C+C_2)}n\varepsilon^2}}{4} \left[2\mathcal A_1(\mF,{\mG})+
\frac{B}{\sqrt 2}\mathcal A_2(\mG)+4\wt C\,B \sqrt{\frac{\log T}{T}}\left(Pdim(\mF)\vee Pdim(\mathcal H)\right)^{1/2} \right] 
\end{align*}
which yields the desired statement.
\end{proof}

\subsection{Proof of Corollary \ref{cor:bound}}\label{sec:proof_cor_bound}

\begin{proof}
We continue to use the shorthand notation $X$ for $\Xn$ and $\wh\b$ for $\wh\b_T$.
Denote with $\wt \pi(\theta)$ the proposal distribution. Then, the posterior $\wt\pi(\theta\C X)$ under $\wt\pi(\theta)$ satisfies
\begin{equation}\label{eq:reweight_posterior}
\pi(\theta\C X)=\wt\pi(\theta\C X)\times R(X)\times r(\theta),\quad\text{where}\quad R(X)= \frac{\wt\pi(X)}{\pi(X)} \,\,\text{and}\,\, r(\theta)=\frac{\pi(\theta)}{\wt\pi(\theta)}.
\end{equation}
Our reconstruction in Algorithm \ref{alg:WGAN-RL} works by first approximating the joint distribution $\wt\pi(\theta,X)$ and then reweighting by the prior ratio, namely
\begin{equation}\label{eq:reweigh}
\pi_{\wh\b}(\theta\C X)=\wt\pi_{\wh\b}(\theta\C X)\times R(X)\times r(\theta),
\end{equation}
where $\wh \b$ has been learned by B-GAN (Algorithm \ref{alg:WGAN}) by matching the joint $\wt\pi(\theta,X)=\wt\pi(\theta\C X)\wt\pi(X)$ under the prior $\wt\pi(\theta)$. We denote the joint measure with this density by $\wt\nu$. 
Denote with $\mu_{\wh\b_T}(X)$ the approximating conditional measure with a density \eqref{eq:reweigh}.  
We can apply the same steps as in the proof of Theorem  \ref{thm:bound} until the step in \eqref{eq:until}. Similarly, we denote with $\wt\E$ the expectation with respect to $\{(\theta_j,X_j)\}_{j=1}^T$ from the ABC reference table sampled from the joint $\tilde \pi(\theta, X)$, and we use $\wt\E_{\wh\b_T}$ interchangeably with $\wt\E$.
The next steps will have minor modifications.  Notice that
$$
\log\frac{\pi(\theta\C X)}{\pi_{\wh\b}(\theta\C X)}=\log\frac{\wt\pi(\theta\C X)}{\wt\pi_{\wh\b}(\theta\C X)}
$$
and thereby
\begin{align*}
&\int_{\mX}\pi(X) {\wt\E_{\wh \b}} d_{TV}^2\left(\nu(X),\mu_{\wh\b_T}(X)\right)\d X\\
&\quad\leq \frac{1}{4}\int_{\mX}\pi(X) {\wt\E_{\wh \b}} \Big[{\KL\big(\nu(X) \| \mu_{\wh\b}(X)\big)+\KL\big(\mu_{\wh\b}(X)\| \nu(X)\big)}\Big]\d X\\
&\quad= \frac{1}{4}  {\wt\E_{\wh \b}} \int_{\mX}\pi(X)\int_\Theta \log\frac{\pi(\theta\C X)}{\pi_{\wh\b}(\theta\C X)}\left[\pi(\theta\C X)-\pi_{\wh\b}(\theta\C X)\right]\d\theta\d X \\
&\quad= \frac{1}{4}  {\wt\E_{\wh\b}}\int_{\mX}\wt\pi( X)\int_\Theta r(\theta) \log\frac{\pi(\theta\C X)}{\pi_{\wh\b}(\theta\C X)}\left[\wt\pi(\theta\C X)-\wt\pi_{\wh\b}(\theta\C X)\right]\d\theta\d X \\
&\quad \equiv \frac{1}{4}  {\wt\E_{\wh\b}} d_{\KL}^S(\wt\nu, \wt\mu_{\wh\b}).
\end{align*}
Using similar arguments as in the proof of Theorem \ref{thm:bound},  we have the following decomposition, for any $\bomega$ such that $f_{\bomega}\in\mF$,
 \begin{align*}
d_{\KL}^S(\wt\nu, \wt\mu_{\wh\b})=& \int_{\mX}\wt\pi(X)\int_\Theta f_\omega(\theta,X)[\wt\pi(\theta\C X)-\wt\pi_{\wh\b}(\theta\C X)]\d\theta \d X\\
&+\int_{\mX}\wt\pi(X) \int_{\Theta}\left[r(\theta) \log \frac{\pi(\theta\C X)}{\pi_{{\wh\b}}(\theta\C X )}-f_\omega(\theta,X)\right][\wt\pi(\theta\C X)-\wt\pi_{{\wh\b}}(\theta\C X)]\d\theta \d X\\
\leq& \,
d_\mF\left(\wt\nu,\wt\mu_{\wh\b}\right)+2\left\|r(\theta) \log \frac{\pi(\theta\C X)}{\pi_{{\wh\b}}(\theta\C X )}-f_\omega(\theta,X)\right\|_{\infty}.
\end{align*}
The rest of the proof is analogous. The only difference is that $\wh\b$ now minimizes the empirical version of $d_\mF\left(\wt\nu,\wt\mu_{\wh\b}\right)$ under the proposal distribution $\wt\pi(\theta)$.
 
\subsection{Motivation for the Sequential Refinement}\label{sec:motivation_RL}

\begin{remark}[2step Motivation]\label{rem:RL}
For the proposal distribution $\wt\pi(\theta)$, using similar arguments as in the proof of Theorem \ref{thm:bound},  the TV distance of the posterior at $X_0$ (not averaged over $P_{\theta_0}^{(n)}$) can be upper-bounded by 
\begin{align*}
4\,d_{TV}^2\left(\nu(X_0),\mu_{\wh\b}(X_0)\right)\leq&\,\, 2\,\mathcal A_1(\mF,\mG,X_0)+ \frac{B}{\sqrt 2}\mathcal A_2(\mG)+4\wt C\,B \sqrt{\frac{\log T\times Pmax}{T}} + A_3(\wt\pi) 
\end{align*}
where $\mathcal A_1(\mF,\mG,X_0)\equiv\sup_{\b:  g_{\b} \in \mG} \inf_{\bomega:f_\bomega \in \mF} \left\|\log \frac{\pi(\theta\C X_0)}{\pi_{{\b}}(\theta\C X_0)}-\frac{f_\omega(X_0,\theta)}{r(\theta)}\right\|$ is the discriminability {\em evaluated at $X_0$} (as opposed to \eqref{eq:term1}) and where
$$
A_3(\wt\pi)=2 \int_{\mX}\wt\pi(X)\left[ \|f_{\bomega}(X_0,\theta)-f_{\bomega}(X,\theta)\|_{\infty}+ B\|g_{\wh\b}(\theta)(X)-g_{\wh\b}(\theta)(X_0) \|_1\right] \d X
$$
and $g_{\wh\b}(\theta)(X)\equiv \pi(\theta\C X)-\pi_{{\wh\b}}(\theta\C X)$.
 This decomposition reveals how the TV distance can be related to discriminability around  $X_0$  and an average discrepancy between the true and approximated posterior densities  relative to their value at $X_0$ where the average is taken over the marginal $\wt\pi(X)$. These averages will be smaller since the marginal $\wt\pi(X)$ produces datasets more similar to $X_0$. For example, an approximation to the the posterior predictive distribution
$\wt\pi(X)=\int_\mX p_{\theta}^{(n)}(X)\pi_{\wh \b}(\theta\C X_0)$ where $\wh\b$ has been learned by B-GAN (Algorithm \ref{alg:WGAN}) is likely to yield datasets similar to $X_0$, thereby producing a tighter upper bound than a flat prior.

\end{remark}

We provide clarifications of the calculations and reasoning in Remark \ref{rem:RL}.
We assume that a prior distribution $\wt\pi(\theta)$ has been used in the ABC reference table that yields the marginal $\wt\pi(X)=\int_{X}p_{\theta}^{(n)}(X)\wt\pi(\theta)\d\theta$. 
Recall the definition of the reweighted posterior reconstruction in \eqref{eq:reweigh} and 
\eqref{eq:reweight_posterior}.
Denote with  
$$
g_{\wh\b}(\theta)(X)\equiv \pi(\theta\C X)-\pi_{{\wh\b}}(\theta\C X)= R(X)\times r(\theta)\times [\wt\pi(\theta\C X)-\wt\pi_{{\wh\b}}(\theta\C X)]
$$ the difference between true and approximated posteriors at $X$, where $\wh\b$ has been trained using the proposal $\wt\pi(\theta)$ and where $R(X)=\wt\pi(X)/\pi(X)$ and $r(\theta)=\pi(\theta)/\wt\pi(\theta)$.
Using similar arguments as in the proof of Theorem \ref{thm:bound},  the squared TV distance of the posterior and its approximation  satisfies, for any element $f_{\bomega}\in\mF=\{f:\|f\|_{\infty}\leq B\}$,
\begin{align*}
4\,d_{TV}^2\left(\nu(X_0),\mu_{\wh\b}(X_0)\right)\leq&  \int_{\Theta} \log \frac{\pi(\theta\C X_0)}{\pi_{{\wh\b}}(\theta\C X_0)}g_{\wh\b}(\theta)(X_0)\d\theta  \\
=& \int_{\Theta}\left[ \log \frac{\pi(\theta\C X_0)}{\pi_{{\wh\b}}(\theta\C X_0)}-\frac{f_\omega(X_0,\theta)}{r(\theta)}\right]g_{\wh\b}(\theta)(X_0)\d\theta \\
&+ \int_{\mX}\pi(X)\int_\Theta \frac{f_\omega(X,\theta)}{r(\theta)}g_{\wh\b}(\theta)(X_0)\d\theta \d X\\
&+ \int_{\mX}\frac{\pi(X)}{r(\theta)}\int_\Theta [f_\omega(X_0,\theta)-{f_\omega(X,\theta)}]g_{\wh\b}(\theta)(X_0)\d\theta \d X\\
\leq &2\, \inf_{\bomega} \left\|\log \frac{\pi(\theta\C X_0)}{\pi_{{\wh\b}}(\theta\C X_0)}-\frac{f_\omega(X_0,\theta)}{r(\theta)} \right\|_{\infty}+ d_\mF\left(\wt\nu,\wt\mu_{\wh\beta}\right)\\
&+2\int_{\mX}\wt\pi(X)\|f_{\bomega}(X_0,\theta)-f_{\bomega}(X,\theta)\|_{\infty}\d X\\
&+2 B\times \int_{\mX}\wt\pi(X)\left\|\frac{g_{\wh\b}(\theta)(X)}{R(X)r(\theta)}-\frac{g_{\wh\b}(\theta)(X_0)}{R(X_0)r(\theta)} \right\|_1\d X.
\end{align*}
The term $d_\mF\left(\wt\nu,\wt\mu_{\wh\b}\right)$ can be bounded as in the proof of Corollary \ref{cor:bound} 
by 
$$
\frac{B}{\sqrt 2}\wt{\mathcal A}_2(\mG)+4\wt C\,B \sqrt{\frac{\log T\times Pmax}{T}} 
$$
Compared to Corollary \ref{cor:bound}, the upper bound on $d_{TV}^2\left(\nu(X_0),\mu_{\wh\b}(X_0)\right)$ now involves the discriminability {\em evaluated at  $X_0$} (not averaged over the marginal $\wt\pi(X)$), i.e.
$$
\mathcal A_1(\mF,\mG,X_0)\equiv\sup_{\b:  g_{\b} \in \mG} \inf_{\bomega:f_\bomega \in \mF} \left\|\log \frac{\pi(\theta\C X_0)}{\pi_{{\b}}(\theta\C X_0)}-\frac{f_\omega(X_0,\theta)}{r(\theta)}\right\|.
$$
The additional two terms in the upper bound involve integration over $\wt\pi(X)$.

\end{proof}

\subsection{Proof of Corollary \ref{cor:dnn}}\label{sec:proof_cor_dnn}

With $\varepsilon_n=O(1/\sqrt{n})$  we need to verify that for some suitable choice  of $C_n\rightarrow\infty$ we have as $n\rightarrow\infty$
\begin{align}
A_1(\mF, \mG)&=o(\e^{-C_n})\label{eq:ass1}\\
A_2(\mG)&=o(\e^{-C_n})\label{eq:ass2}\\
 {\frac{\log T}{T}}\times \big[Pdim(\mF)\vee Pdim(\mathcal H) \big]&=o(\e^{-2C_n})\label{eq:ass3}
\end{align}
for $T$ that is large enough, i.e. $T\geq Pdim(\mF)\vee Pdim(\mathcal H)$.
The term $\mathcal A_2(\mG)$ equals zero from our assumption of representability of $\pi(\theta,X)=\pi_{g_{\b_0}}(\theta,X)$ for $g_{\b_0}\in\mathcal G$, which verifies \eqref{eq:ass2}.
We assume that $\Xn$ is a stacked vector of $n$ observed vectors of length $q$, not necessarily iid, and denote  $d^*=d+nq$. 
Using leaky ReLU networks and assuming representability,  for any $\b$ such that $g_{\b}\in\mathcal G$  the log posterior ratio $r_{\b}(\theta,\Xn)=\log \frac{\pi(\theta\C \Xn)}{\pi_{\b}(\theta\C \Xn)}$ is continuous and, due to boundedness of the network weights, satisfies
$$
0<\underbar C\leq r_{\b}(\theta,\Xn)\leq \bar C<\infty
$$
for any fixed $d^*$. With large enough $T$ and setting $E=[-\log T,\log T]^{d^*}$  and $R=\log T$, Lemma \ref{lemma:approx} yields that there exists a ReLU network $f_{\bomega}\in\mF $ 
with a width 
$$
W=3^{d^*+3}\max\{d^*\lfloor N^{1/d^*}\rfloor, N+1\}
$$ 
and depth $L=12\log T+14+2d^*$ such that 
{$$
\mathcal A_1(\mF,\mG)\leq {\sup_{\b: g_{\b} \in \mG } \inf_{\bomega: f_{\bomega}\in \mF}} \|r_{\b}(\theta,\Xn)- f_{\bomega}(\Xn,\theta)\|_{L_\infty(E)}\leq 19 \sqrt{d^*}\omega_f^E(2 (\log T)^{1-2/d^*}N^{-2/d^*}), 
$$}
where $\omega_f^E$  is the modulus of continuity of $f(t)$ satisfying $\omega_f^E(t)\rightarrow 0$ as $t\rightarrow0^+.$
Choosing $N$ such that $2^{d^*/2} (\log T)^{d^*/2-1}=o(N)$ as $T\rightarrow\infty$, the right-hand side above goes to zero for any fixed $d^*=d+nq$. For each $n$, we can find $T$ large enough (depending on the modulus of continuity) 
such that $\mathcal A_1(\mF,\mG)\e^{C_n}\sqrt{d^*}\leq \eta_n$ for some $\eta_n=o(1)$, yielding \eqref{eq:ass1}. The smallest $T$ that satisfies this will be denoted with $T_n$.

In order to verify \eqref{eq:ass3},   Theorem 14.1 in  \citet{anthony1999neural} and Theorem 6 in \citet{bartlett2017spectrally}  show that for piecewise linear activation functions (including ReLU and leaky ReLU) there exist constants $c,C>0$ such that
$$
c\times SL\log (S/L)\leq Pdim(\mF)\leq C\times S L\log S,
$$
where $\mF$ is a class of discriminator networks with $L$ layers and $S$ parameters. 
Since elements in $\mathcal H$ can be regarded as sparse larger neural networks with $L+L^*$ layers, $S+S^*$ parameters and piece-wise linear activations, we have
$$
Pdim(\mF)\vee Pdim(\mathcal H) \leq C\times (S+S^*) (L+L^*)\log (S+S^*).
$$
Our assumption $T\geq Pdim(\mF)\vee Pdim(\mathcal H)$ will thus be satisfied, for instance, when 
\begin{equation}\label{eq:T}
T>C\times (S+S^*) (L+L^*)\log (S+S^*).
\end{equation}
Choosing $N=\lfloor 2^{d^*/2} (\log T)^{d^*/2}\rfloor$, which satisfies the requirement  $2^{d^*/2} (\log T)^{d^*/2-1}=o(N)$ as $T\rightarrow\infty$, yields 
\begin{equation}\label{tune:W}
W=3^{d^*+3}\max\{d^*\lfloor N^{1/d^*}\rfloor, N+1\}=3^{d^*+3}\lfloor 2^{d^*/2} (\log T)^{d^*/2}+1\rfloor
\end{equation}
for a sufficiently large $n$ (and thereby $d^*$).
Recall that in the feed-forward neural networks, the total number of parameters $S=\sum_{l=0}^{L-1}[w_l(w_l+1)]$ satisfies $S\leq LW(W+1)$. 
For any fixed $d^*$ (and thereby $n$), assuming $L=12\log T+14+2d^*$ as before and $W$ as in \eqref{tune:W}, we define $T(d^*)$ as the smallest $T$ that satisfies
$$
C\times [LW(W+1)+S^*] (L+L^*)\log [LW(W+1)+S^*]\leq \frac{T}{\log T}\times\e^{-2C_n}\times \eta_n
$$
for some $\eta_n=o(1)$. Any $T>T(d^*)$ satisfies $T>Pdim(\mF)\vee Pdim(\mathcal H)$ and $\e^{2C_n}\frac{\log T}{T}[ Pdim(\mF)\vee Pdim(\mathcal H)]\leq \eta_n$. With $T\geq \max\{T_n,T(d^*)\}$, the condition \eqref{eq:ass3} is verified.

\begin{lemma}\citep[Lemma B5]{zhou2022deep}\label{lemma:approx}
Let $f$ be a uniformly continuous function defined on $E\subset[-R,R]^d$. For any $L,N\in\mathbb N^+$, there exists a ReLU network function $f_{\bphi}$ with width $3^{d+3}\max\{d\lfloor N^{1/d}\rfloor, N+1\}$ and depth $12L+14+2d$ such that
$$
\|f-f_{\bphi}\|_{L_{\infty}(E)}\leq 19\sqrt{d} \omega_f^E(2RN^{-2/d}L^{-2/d}),
$$
where $\omega_f^E(t)$ is the modulus of continuity of $f(t)$ satisfying $\omega_f^E(t)\rightarrow 0$ as $t\rightarrow0^+.$
\end{lemma}

\subsection{Theory for Adversarial Variational Bayes}\label{sec:thm_VB}

\begin{thm}\label{thm:vb_bound}
Let $\wh \b_T$ be as in  \eqref{eq:psihat} where $\mF=\{f: \|f\|_\infty\leq B\}$ for some $B>0$. Denote with $\E$ the expectation with respect to  $\{Z_j\}_{j=1}^T$ in the   reference table. Assume that  the prior satisfies \eqref{eq:prior_conc}.
Then for $T\geq Pdim(\mF\circ\mG)$ we have for any $C_n>0$
$$
P_{\theta_0}^{(n)}\E\,  d_{TV}^2\left(\nu(X_0),\mu_{\wh\b_T}(X_0)\right)\leq \mathcal D_n^T({\mF,\mG},\varepsilon_n,C_n),
$$
where 
$$
D_n^T({\mF,\mG},\varepsilon_n,C_n)=\frac{\mathcal A_3({\mF,\mG})}{2}+\frac{1}{C_n^2n\varepsilon_n^2}+\frac{1}{2}\wt C\sqrt{\frac{\log T}{T}  Pdim(\mF\circ\mathcal G)}+\frac{\e^{ (1+C_2+C_n)n\varepsilon_n^2}}{4} \frac{B}{\sqrt{2}}\mathcal A_2(\mG)
$$ 
for some $\wt C>0$ where $\mathcal A_2(\mG)$ was defined in \eqref{eq:term2} and where 
$$
\mathcal A_3({\mF,\mG})=  P_{\theta_0}^{(n)} {\E} \left\| \log \frac{\pi_{\wh\b_T}(\theta\C X_0)}{\pi(\theta\C X_0)}- f_{\bomega(\wh\b_T)} (X_0,\theta)\right\|_\infty.
$$
{Here $\E$ account for the nested randomness in the estimation process of $\bomega(\wh\b_T)$ and $\hat \b_T$.}
\end{thm}   
\proof
We denote with $\bar E$ the expectation with respect to the empirical distribution. 
Because   the class $\mF$ is symmetrical (i.e. $f\in\mF$ implies $-f\in\mF$), the adversarial variational Bayes estimator is defined as 
\begin{equation}\label{eq:psihat}
\wh\b_T=\arg\min\limits_{\b:g_{\b}\in\mathcal G}\left[  \bar E_{Z\sim\pi_Z} f_{\bomega(\b)}\big(X_0, g_{\b}(Z,X_0)\big)-   E_{\theta\sim\pi(\theta\C X_0)} f_{\bomega(\b)}\left(X_0,\theta\right)\right] 
\end{equation}
where
$$
\bomega(\b)=\arg\max\limits_{\bomega:f_{\bomega\in\mF}} \left[ \bar E_{Z\sim\pi_Z,X\sim\pi(X)} f_{\bomega}\left(X,g_{\b}(Z,X)\right) -\bar E_{(\theta,X)\sim\pi(\theta,X)} f_{\bomega}(X,\theta)\right].
$$
Note that the (stochastic) gradient descent update for $\b$, conditioning on the most recent value of $\bomega$, {\em does not} involve the second term $E_{\theta\sim\pi(\theta\C X_0)} f_{\bomega(\b)}\left(\theta,X_0\right)$ in \eqref{eq:psihat} because it does not depend on $\b$. The minimization occurs only over the first term. We obtain theoretical results for $\wh\b_T$ and note that our Algorithm  \ref{alg:WGAN-VB} targets this estimator. In the sequel, we denote $\wh\b_T$ simply by $\wh\b$ and use the notation $\pi_{\b}(\theta,X)=\pi_{\b}(\theta\C X)\pi(X)$ for the joint generator model.
 Using the Pinsker inequality we obtain 
\begin{align*}
P_{\theta_0}^{(n)} {\E}4\,d_{TV}^2\left(\nu(X_0),\mu_{\wh\b}(X_0)\right)\leq&  P_{\theta_0}^{(n)}  {\E} \int_{\Theta} \log \frac{\pi(\theta\C X_0)}{\pi_{{\wh\b}}(\theta\C X_0)}\big[\pi(\theta\C X_0)-\pi_{\wh\b}(\theta\C X_0)\big]\d\theta  \\
\leq & P_{\theta_0}^{(n)}  {\E} 2\left\| \log \frac{\pi_{\wh\b}(\theta\C X_0)}{\pi(\theta\C X_0)}- f_{\bomega(\wh\b)} (X_0,\theta)\right\|_\infty\\
&+ P_{\theta_0}^{(n)}  {\E} d_{\wh\b}(\nu_{\wh\b}(X_0), \mu(X_0)), 
\end{align*}
where we define (for any   $\b,\wt\b$ such that $g_{\b}\in\mG$ and $g_{\wt\b}\in\mG$)
$$
d_{\wt\b}(\nu_{\b}(X),\mu(X))\equiv   E_{\theta\sim\pi_{\b}(\theta\C X)}f_{\bomega(\wt\b)} (X,\theta)-  E_{\theta\sim\pi(\theta\C X)}f_{\bomega(\wt\b)} (X,\theta).
$$
From the definition of \eqref{eq:psihat} and since $\mathcal F\circ \mathcal G$ is symmetrical, we have for any realization $X_0$ and  for any $\b$
\begin{align}
d_{\wh\b}(\nu_{\wh\b}(X_0),\mu(X_0))&=  d_{\wh\b}(\nu_{\wh\b}(X_0),\bar \nu_{\wh\b}(X_0))+ d_{\wh\b}(\bar \nu_{\wh\b}(X_0),\mu(X_0))\notag\\
&\leq   d_{\wh\b}(\nu_{\wh\b}(X_0),\bar \nu_{\wh\b}(X_0))+ d_{\b}(\bar \nu_{\b}(X_0),\mu(X_0))\label{eq:def_min}\\
&\leq   d_{\mF}(\nu_{\wh\b}(X_0),\bar \nu_{\wh\b}(X_0))+ d_{\b}(\bar \nu_{\b}(X_0),  \nu_{\b}(X_0)) + d_{\b}(\nu_{\b}(X_0),\mu(X_0))\notag\\
&\leq  2d_{\mF\circ\mathcal G}(\pi_Z,\bar\pi_Z)+  d_{\b}(\nu_{\b}(X_0),\mu(X_0)).\notag
\end{align}
Next, using the same arguments as in the proof of Theorem \ref{thm:bound}, we obtain
\begin{align*}
P_{\theta_0}^{(n)} {\E}   d_{\b}(\nu_{\b}(X_0),\mu(X_0))&\leq 4\,\P_{\theta_0}^{(n)}[\mathcal A_n^c(\epsilon)]+\e^{(1+C_n+C_2)n\epsilon^2}  {\E} \int_{\mX} d_{\b}(\nu_{\b}(X),\mu(X))\pi(X)\d X 
\end{align*}
 Since $\|f_{\bomega(\b)}\|_\infty\leq B$, we  have for any $\b$
\begin{align*}
E_X d_{ \b}(\mu_{\b}(X),\nu(X))&= \int_{\mX}\pi(X)\int_{\Theta}f_{\bomega(\b)} (X,\theta)\big[\pi(\theta\C X)-\pi_{\b}(\theta\C X)\big]\d\theta\d X \\
&\leq B\times E_X  d_{TV}(\mu_{\b}(X),\nu(X))\\
&\leq \frac{B}{\sqrt{2}}E_X\sqrt{\left\|\log\frac{\pi(\theta\C X)}{\pi_{\b}(\theta\C X)}\right\|_{\infty}}.
\end{align*}
In the sequel, we choose $\b$ which minimizes this term.
Next, using the symmetrization techniques as before in the proof of Theorem \ref{thm:bound} and denoting with $\E$ the expectation with respect to $\{Z_j\}_{j=1}^T$, we have
\begin{align*}
 d_{\mF\circ\mathcal G}(\pi_Z,\bar \pi_Z)&\leq  \,\E \mathcal R_n(\mF\circ\mathcal G) \leq \wt C\sqrt{\frac{\log T}{T}  Pdim(\mF\circ\mathcal G)}.
\end{align*}
Putting the pieces together, we obtain
an upper bound $D_n^T( {\mF, \mG} ,\varepsilon_n,C_n)$ for $P_{\theta_0}^{(n)}\E\,  d_{TV}^2\left(\nu(X_0),\mu_{\wh\b_T}(X_0)\right)$.

\section{More on Adversarial Variational Bayes}\label
{sec:supplement_VB}
The VB algorithm reported in the main paper follows the scheme in  \Cref{alg:WGAN-VB}. \Cref{fig:gauss_ex2_another} includes an example of performance of Algorithm \ref{alg:WGAN-VB} on another realization of Example \ref{ex:toy}. 

\begin{figure}[!ht] 
\centering
\includegraphics[width=0.95\textwidth]{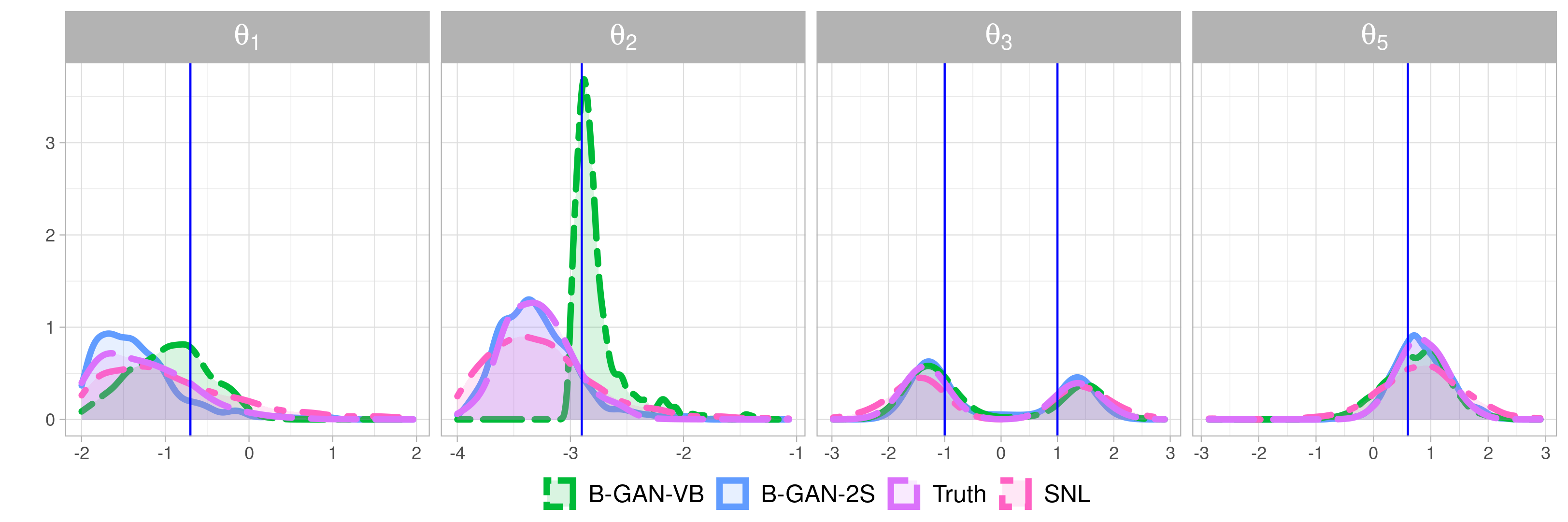}
\caption{Posterior densities under the Gaussian model (another repetition). The true parameter is $\btheta_0=(-0.7, -2.9, -1.0, -0.9, 0.6)'$, while the signs of $\theta_3$ and $\theta_4$ are not identifiable.}\label{fig:gauss_ex2_another}
\end{figure}

\section{Jensen-Shannon Version of B-GAN}\label{sec:JS_cGAN}

The empirical version of the minimax game  using \eqref{eq:gan_obj}  involves the ABC reference table $\{(\theta_j,X_j)\}_{j=1}^T$  and  noise realizations $\{Z_j\}_{j=1}^T \sim \pi_Z(\cdot)$
to solve
\begin{equation}\label{eq:gan2}
(g^*,d^*)=\arg\min_{g\in\mathcal G}\max_{d\in\mathcal D }\left(\sum_{j=1}^T  \log d(X_j,\theta_j)+\sum_{j=1}^T \log\Big[1-d\big(X_j,g(Z_j,X_j)\big)\Big]\right).
\end{equation}

The convergence difficulties of  Jensen-Shannon (JS) GANs \citep{goodfellow2014generative}, similar to  \Cref{alg:KL-GAN},  have been no secret. We provide a simple illustration of how it can fail on the toy example in \Cref{ex:toy}. In particular, we show that the convergence is very sensitive to the choice of the learning rate (step size in stochastic gradient descent).  To make comparisons with the Wasserstein version more fair, we use the same network architecture  used for Algorithm \ref{alg:WGAN} described in \Cref{sec:implementation_toy}  for the JS version as well. At the end of this section we  point out that  JS can work with more careful tuning at higher computational cost.

\begin{figure}[!ht]
\centering
\begin{minipage}{.78\linewidth}
\begin{algorithm}[H]
	\caption{\em B-GAN   (Jensen-Shannon Version)}\label{alg:KL-GAN}
\centering
	\spacingset{1.1}
	\small
	\resizebox{\linewidth}{!}{
		\begin{tabular}{l l}
			\hline 
			\multicolumn{2}{c}{\bf INPUT \cellcolor[gray]{0.6} }\\	
				\hline
			\multicolumn{2}{c}{Prior $\pi(\theta)$, observed data $X_0$ and noise distribution $\pi_Z(\cdot)$}\\
				\hline
			\multicolumn{2}{c}{\bf PROCEDURE \cellcolor[gray]{0.6}}\\		
			\hline
			\multicolumn{2}{c}{Initialize  networks $d^{(0)}$ and $g^{(0)}$}\\
			\multicolumn{2}{c}{\bf  \cellcolor[gray]{0.9}ABC Reference Table}\\
						For $j=1,\dots, T$:  
                       Generate   $(X_j,\theta_j)$ where $\theta_j\sim\pi(\theta)$ and $X_j\sim P_{\theta_j}^{(n)}$.  \\
			\multicolumn{2}{c}{\bf  \cellcolor[gray]{0.9}GAN Training }\\
			\multicolumn{2}{l}{For $t=1,\dots, N$:}\\
			\qquad  Generate noise $Z_{j}\sim\pi_Z(z)$ for  $j=1,\ldots, T$. &\\
			\qquad  Update $d^{(t)}$ and $g^{(t)}$  using stochastic gradient descent applied to  \eqref{eq:gan2}. &\\
			\multicolumn{2}{c}{\bf  \cellcolor[gray]{0.6}POSTERIOR SIMULATION}\\
			\multicolumn{2}{l}{For $i=1,\dots, M$: 	Simulate	 $Z_i\sim\pi_Z(z)$ and set $\wt \theta_i=  g^{(N)}(Z_i, X_0^{(n)})$}.
				\end{tabular}}
\end{algorithm}
\end{minipage}
\end{figure}

One of the most common causes of failure is an overly-powerful discriminator. Because of the  `log loss', when the discriminator learns too fast and becomes too strong, the `close-to-boundary' predictions (near 0 for fake data or near 1 for real data) cause vanishing gradients for the generator.
We show how the training balance (between the generator and the discriminator) can be easily disturbed when we alter the learning rate of the generator. We consider two scenarios where the noise variables $\{Z_j\}_{j=1}^T$ are  (a) refreshed for each stochastic gradient step, and (b) when they are sampled only once ahead of time and then random minibatches selected for each stochastic gradient steps. In \Cref{fig:KL_theta}, we report the four approximated posteriors for these two scenarios considering two learning rates of the generator (weak $\text{lr}_g=10^{-4}$ versus strong $\text{lr}_g=10^{-3}$) while  keeping the learning rate of the discriminator at $10^{-4}$. When the generator is weak compared to the discriminator ($\text{lr}_g=10^{-4}$), the posterior reconstructions are not quite around the true values (bottom figure). When the generator is stronger ($\text{lr}_g=10^{-3}$), the posteriors at least cover the correct locations of the parameters, but they are not nearly as successful as the Wasserstein  B-GAN reconstructions we have seen in \Cref{fig:gauss_ex1}. It is also interesting to compare the results for the two treatments of the noise $Z$.  Sampling $Z$'s ahead of time and then sub-sampling for stochastic gradient yields less satisfactory reconstructions. 

\begin{figure}[!t]
\centering
\subfigure[$\text{lr}_g=10^{-3}$]{\includegraphics[width=\textwidth,height=0.25\textwidth]{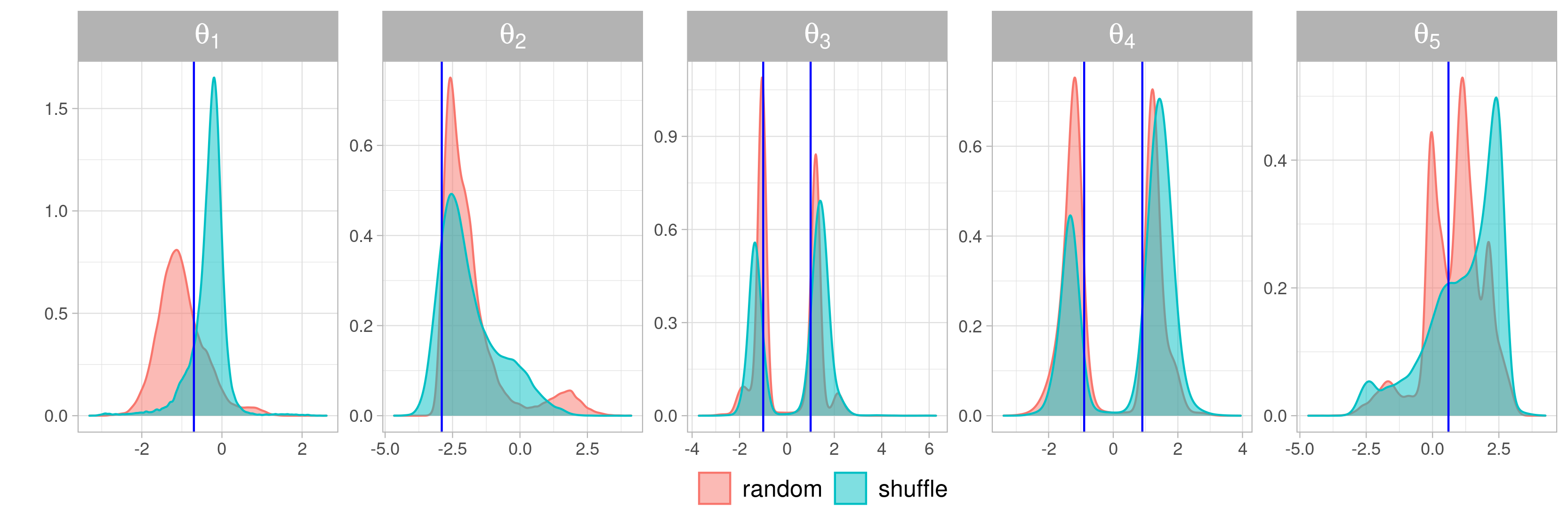}}
\subfigure[$\text{lr}_g=10^{-4}$]{\includegraphics[width=\textwidth,height=0.25\textwidth]{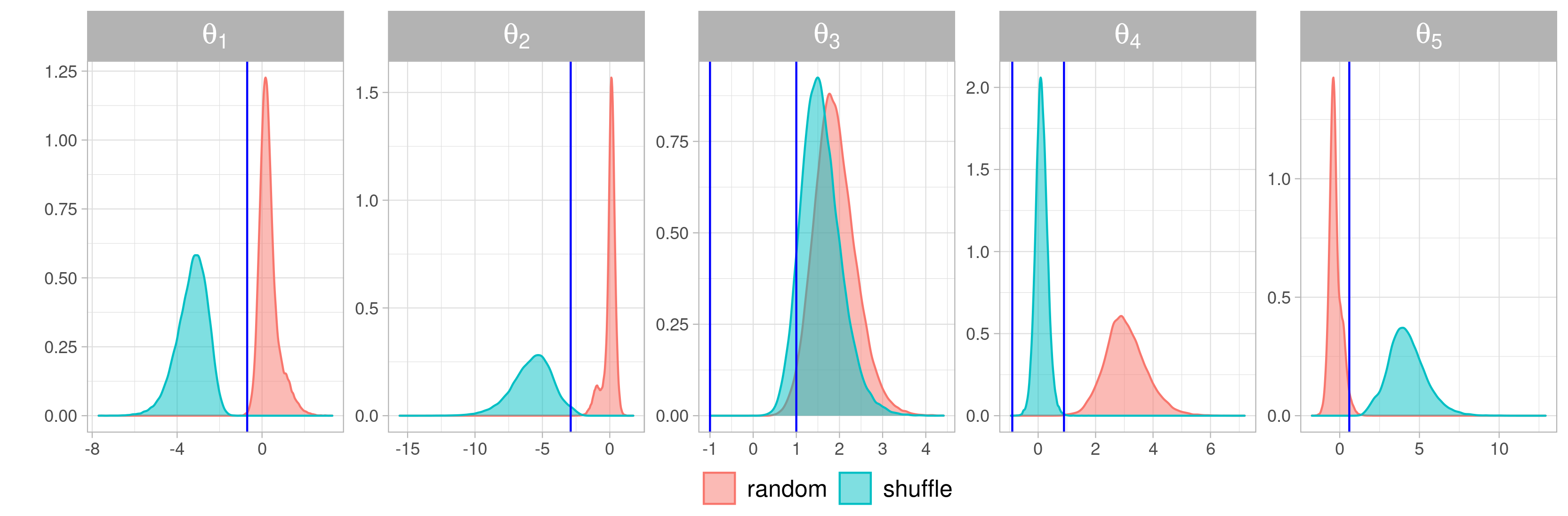}}
\caption{\footnotesize Approximated posteriors under different generator learning rates and different $Z$ randomizations. The learning rate of the discriminator is fixed at $10^{-4}$ in both. The blue vertical lines mark the correct locations of the parameters. }\label{fig:KL_theta}
\end{figure}

 The JS version of  B-GAN could work with a more delicate  calibration and extensive training. 
 The Gaussian example was also studied in \citet[Section 3.1]{ramesh2022gatsbi}, who constructed much bigger and deeper networks for both the discriminator and generator. They used a 4-layer network for the generator, with   layer widths equal to $(d_X+d_Z, 128, 128,  128,  128, d_\theta)$, and a 5-layer network for the discriminator, with layer widths $(d_X+d_\theta, 2048, 2048, 2048, 2048, 2048, 1)$. They used  leaky ReLU activations with a $0.1$ negative slope. Using $100\,000$ pairs of $(X_i, \theta_i)$ (twice as many compared to what we used), the networks were trained for $20\,000$ epochs with $100$ discriminator updates per each generator update. Spectral normalization is also applied to ensure stable training. Yet, the authors observed that the performance on the Gaussian example is inferior to SNL (Figure 2A in their paper). Our Wasserstein B-GAN implementation could outperform SNL with much simpler networks, smaller ABC reference table and significantly lower optimization costs. We explored other  structures, which are different from the ones mentioned above and simpler than the ones used by the authors. For example, a generator network with input-output dimensionality as $(d_X+d_Z, 128, 128, 128, d_\theta)$ and a discriminator network with input-output dimensionality as $(d_X+d_\theta, 512, 512, 512, 512, 1)$. However, they did not produce satisfactory posterior approximations.
 
\section{Implementation}\label{sec:implementation}

\subsection{Implementation details of the one-sided gradient penalty}
For the Lipschitz constraint in \eqref{eq:wgan_obj}, to avoid the ad hoc gradient clipping in \citet{arjovsky2017towards}, \citet{gulrajani2017improved} proposed the soft constraint with a penalty on the gradient of $f_\omega$ with respect to a convex combination of the two constrasting datasets. They adopt the two-sided penalty encouraging the norm of the gradient to go towards $1$ instead of just staying below 1 (one-sided penalty). This is inspired by the fact that the optimal critic function contains straight lines with gradient norm $1$ connecting coupled points from the contrasted distributions \citep[Proposition1]{gulrajani2017improved}. Similarly as \citet{athey2021using}, we adopt the one-sided penalty only with respect to $\theta_j$  
 \begin{equation}\label{eq:gp}
\lambda \Big\{ \frac{1}{T}\sum_{j=1}^T \Big[\max\Big(0, \norm{\nabla_{\bar \theta} f_{\bomega}(X_j,\bar \theta_j)}_2-1\Big)\Big]^2 \Big\}
\end{equation}
where $\bar \theta_j = \epsilon_j \theta_j + (1-\epsilon_j)g(Z_j, X_j)$ with the $\epsilon_j$ re-drawn from a uniform distribution at each step. The choice of $\lambda$ is discussed in \Cref{sec:implementation}.

\subsection{Implementation Details for the Gaussian Example} 
\label{sec:implementation_toy}

\begin{table}[!ht]
\centering
\scalebox{0.8}{
\begin{tabular}{l | *{10}{c}}
\toprule
        & \multicolumn{2}{c}{$\theta_1=0.7$} & \multicolumn{2}{c}{$\theta_2=-2.9$} & \multicolumn{2}{c}{$\abs{\theta_3}=1.0$} & \multicolumn{2}{c}{$
        \abs{\theta_4}=0.9$} & \multicolumn{2}{c}{$\theta_5=0.6$} \\
        & bias       & CI width     & bias       & CI width     & bias       & CI width      & bias       & CI width     & bias       & CI width      \\
                \midrule
Truth   & 0.63 & 1.50 (0.8) & 0.39 & 1.04 (0.9) & 0.37 & 2.85 & 0.27 & 2.39 & 0.83 & 1.68 (0.8) \smallskip \\
SNL     & 0.73 & 3.50 & 0.49 & 2.38 & 0.40 & 4.11 & 0.28 & 3.47 & 0.88 & 3.33 \\
SS      & 1.25 & 5.35 & 1.48 & 6.53 & 0.83 & 5.66 & 0.86 & 5.67 & 1.56 & 5.69 \\
W2      & 1.14 & 3.74 & 0.98 & 3.29 (0.9) & 0.46 & 4.20 & 0.45 & 3.70 & 1.36 & 5.52 \smallskip\\

\hline 

& \multicolumn{10}{c}{ReLU activations} \\

B-GAN        & 0.69 & 2.89 & 0.43 & 2.08 & 0.37 & 3.54 & 0.29 & 3.12 & 0.87 & 2.38 (0.8)  \smallskip  \\

&\multicolumn{10}{c}{Architecture $(128,128,128)$}\\
B-GAN-2S & 0.57 & 1.86 (0.8) & 0.32 & 1.15 (0.9) & 0.31 & 2.91 (0.9) & 0.23 & 2.42 (0.9) & 0.84 & 2.16 (0.8)   \\

&\multicolumn{10}{c}{Architecture $(256, 256)$}\\
B-GAN-2S &  0.57 & 1.73 (0.8) & 0.32 & 1.18 (0.9)& \bf{0.28} & \bf{2.76 (0.8)} & 0.20 & \bf{2.30 (0.9)} & 0.73 & 1.79 (0.8) \\

\hline
& \multicolumn{10}{c}{Leaky ReLU activations} \\

B-GAN    & 0.64 & 2.74 & 0.39 & 1.92 & 0.35 & 3.53 & 0.26 & 3.04 & 0.86 & 2.28  (0.8)  \\
B-GAN-2S & \bf{0.53} & 1.77 (0.8) & 0.30 & \bf{1.12 (0.9)} & 0.33 & 2.98 & 0.27 & 2.52 & 0.83 & 2.12  (0.9)  \\
\bottomrule   
\end{tabular}}
\caption{\footnotesize Summary statistics of the approximated posteriors under the Gaussian model (averaged over 10 repetitions).  For $\theta_3$ and $\theta_4$, we compute the statistics using the absolute values of the parameters, since the posteriors have symmetric modes. Truth refers to the posterior calculated from the exact likelihood function. Bold fonts mark the best model of each column (excluding the true posterior). The coverage of the intervals are 1 unless otherwise stated in the parenthese.}\label{tab:gauss_post}
\end{table}

 \paragraph{Network Architectures.} Our implementation of Algorithm \ref{alg:WGAN} in python builds on the codes provided by \citep{athey2021using}.\footnote{https://github.com/evanmunro/dswgan-paper} We use fairly modest generator/critic networks. For the generator network, we use only 3 hidden layers totaling in dimensions   $(d_Z+d_X,128, 128,128, d_\theta)$. For the critic network, we use a similar architecture with layer dimensions equal to $(d_\theta+d_X,128, 128, 128, 1)$.  A small amount of regularization (dropout $=0.1$ for each layer) was applied  to avoid over-fitting. All the weights are initialized at $0$ and the $\pi_Z(\cdot)$ is specified as the mean zero Gaussian with identity covariance matrix, dimension $d_Z$ equal to $d_\theta$.

  For the post-processing enhancements in \Cref{alg:WGAN-RL} and \Cref{alg:WGAN-VB}, we choose shallower but wider networks with the hope that they will better capture local aspects \citep{chen1999improved}. In particular, for both the generator and critic networks, we use only two-layer networks, resulting in layer dimensions   $(d_Z+d_X, 256, 256, d_\theta)$ for the generator and $(d_\theta+ d_X, 256, 256, 1)$ for the critic. We summarize comparisons of this setting with the previous 3 layers of 128 units for B-GAN-2S in \Cref{tab:gauss_post}. From the results, we see that  wider and shallower networks are superior to the deeper and narrower ones for the two local refinements. We have thereby reported posteriors with two layers of 256 units networks in the main paper. In terms of activation functions, our previous analyses are conducted with a ReLU activation, given its expressibility and inherent sparsity \citep{goodfellow2016deep}. We have also considered the leaky ReLU with a negative slope $a=0.1$ (see \Cref{tab:gauss_post}) and {found no significant advantage of one versus the other.} We report results for ReLUs in \Cref{fig:gauss_ex1}.

\paragraph{Hyperparameters.} Regarding the  choice of $T$, for our ABC methods (with summary statistics and the Wasserstein distance), we use a reference table of size $T=100\,000$. We construct approximated posteriors using the top $1\%$ draws with the smallest ABC discrepancies. For Algorithm \ref{alg:WGAN},  we use the same reference table. However,  for the stochastic gradient descent updates we use a batchsize $T=6400$ (implementations in \citep{ramesh2022gatsbi} suggest a batch size around $10\%$ of the total sample size $T$). The data pairs $(X_j,\theta_j)$ are thus subsetted with replacement (not re-sampled) for each iteration of the stochastic gradient descent.   The noise variables $Z_j$'s, however,  are refreshed (not pre-sampled and subsetted) for each batch.  This is commonly used in existing GAN implementations, including
\citep{athey2021using}. 

   We set $n_\text{critic} =15, \lambda=5$, and the learning rate for the two networks as $\text{lr}_g=\text{lr}_c=10^{-4}$, which are used in both \citep{gulrajani2017improved}  and \citep{athey2021using}.  For optimization, we use the ADAptive Moment estimation algorithm (ADAM) by \citep{kingma2014adam}.  We train B-GAN for $N=1\,000$ epochs or until convergence (the test loss stops decreasing).

For Algorithm \ref{alg:WGAN-RL} and \ref{alg:WGAN-VB}, we use smaller reference tables ($T=50\,000$) with a smaller batch-size $1280$, again training the networks for $N=1000$ epochs. We increase the $n_\text{critic}$ and the penalty $\lambda$ with the hope that a regularized and  stabilized critic could help the generator to learn better in the local region. 
 In general, a well-behaved critic network always helps but the extra training costs could be high when the sample size is large.
 Since these local enhancement variants are trained on a smaller reference table, we make these alterations so that they can converge faster with only a minor increase in computation costs.  
In addition, unlike the JS version of GAN as described in 
\Cref{alg:KL-GAN}, the Wasserstein version is less sensitive to the choice of hyper-parameters. Essentially, different hyper-parameters yield similar results and only lead to a trade-off between the convergence speed and computation costs.

\paragraph{Details of Other Methods.} For the SNL model \citep{papamakarios2019sequential}, we adopt the configurations suggested by the authors.\footnote{The authors provide their implementations on https://github.com/gpapamak/snl} They generated 1\,000 pairs of $(\theta_j, X_j)$ in each round, with 5\% randomly selected to be used as a validation set. They stopped training if the validation log likelihood did not improve after 20 epochs.  They suggested 40 rounds for  the Gaussian model.
Each Masked Autoregressive Flow (MAF) network has two layers of 50 hidden units with hyperbolic tangent function (tanh) activations.   

For ABC methods, we use the mean and variance as summary statistics for naive ABC (SS) and the Wasserstein version is implemented using the 2-Wasserstein (W2) distance under the Euclidean metric defined as $W2(X_i, X_j) =\min_\gamma \big[\sum_{k_1=1}^n \sum_{k_2=1}^n \gamma_{k_1 k_2}\norm{X_{i, k_1} - X_{j, k_2}}^2 \big]^{1/2}$ s.t. $\gamma' \1_n=\1_n, \gamma \1_n=\1_n$ with $0\leq \gamma_{k_1 k_2}\leq 1$. Note that when we calculate the 2-Wasserstein distance, each $X_i$ is treated as 4 pairs $(n=4)$ of bi-variate normal variables rather than a flat vector of length 8.  For both methods, the approximated posteriors are constructed using the ABC draws with the top 1\% smallest  data discrepancies. Performance details (after 10 repetitions) are summarized in \Cref{tab:gauss_post}.

\subsection{Implementation Details for the I.I.D. M/G/1-Queuing Model}\label{app:iid_implementation}
For the M/G/1-queuing example, we generate ABC reference table of size $T=100\,000$ in all three examples. For the stacking and deep set structure, the entire table is generated upfront as $\{(\theta_i, \Xn_i)\}$. For the stacking implementation, $\Xn_i$ is flatten as a long tensor of shape $(1, n)$, while the deep set structure takes a three-dimensional tensor of shape $(1, n^*, 5)$. For the sequential update implementation, we generate new ABC reference table with $T=100\,000$  as $\{\theta_i, X^{n/5}_i\}$ (we use 5 batch update) with $\theta_i$ generated from the posterior distribution learned from last iteration and each dataset only contains $n^*/5$ i.i.d. observations and is flatten of shape $(1, n/5)$.

For our deep set architecture for learning the conditional distribution on i.i.d. datasets, it is structured as
\begin{align*}
g_\beta(Z, [X_1, X_2, \ldots, X_n]) &= g^{(1)}_{\beta_1} (Z, \sum_{i=1}^n g^{(2)}_{\beta_2}(Z, X_i))\\
f_\omega(\theta, [X_1, X_2, \ldots, X_n]) &=  f^{(1)}_{\omega_1} (\theta, \sum_{i=1}^n f^{(2)}_{\omega_2}(\theta, X_i))
\end{align*}
where $g^{(1)}_{\beta_1},  g^{(2)}_{\beta_2}, f^{(1)}_{\omega_1}, f^{(2)}_{\omega_2}$ are sub-network structures inside the generator and critic functions. For the M/G/1-queuing example, we use the same deep set structure for $g_{\beta_2}^{(2)}$ and $f_{\omega_2}^{(2)}$ as ReLU networks with $1$ hidden layers with $64$ units. We use 2 hidden layers with $(64, 64)$ units for $g^{(1)}_{\beta_1}$ and 1 hidden layer with 64 units for $f^{(1)}_{\omega_1}$. We set $\text{lr}_g=1e^{-3}$ and $\text{lr}_c=1e^{-4}$, and train for 2\,000 epoches with batch size of 1\,280.

For the sequential implementation, we use 3 hidden ReLU layers of $(128, 128, 128)$ units for both the generator and the critic functions. We set $\text{lr}_g=5e^{-4}$ and $\text{lr}_c=1e^{-4}$ and train the networks with batch size of 1\,280.  We train the network for 1\,000 epoches for the first batch and 500 epoches for the remaining batches.

For the stacked implementation, we use 3 hidden ReLU layers of $(128, 128, 128)$ units for both the generator and the critic functions. We set $\text{lr}_g=5e^{-4}$ and $\text{lr}_c=1e^{-4}$ and train the networks for 2\,000 epoches with batch size of 1\,280.

\subsection{Implementation Details for the Lotka-Volterra Example}\label{sec:implementation_lv}

For the Lotka-Volterra example, we adopt the 9-dimensional summary statistics used in \citet{papamakarios2016fast}. For \Cref{alg:WGAN}, we use ReLU neural networks with $L-1=3$ hidden layers with $(128,128,128)$ units from bottom to top for both the generator and the critic functions. We generate $T=1\,000\,000$ pairs of $(X_i, \theta_i)$ for the vanilla B-GAN, and the networks are trained with a batch size $B=12\, 80$ for $1\,000$ epochs.  We adopt the same $n_\text{critic}, \lambda, \text{lr}_g =\text{lr}_c$ as the Gaussian example in \Cref{sec:implementation_toy}.

 For learning from the adjusted prior $\wt\pi(\theta)=\pi_{\wh g}(\theta\C X_0)$ where $\wh g$ was obtained from \Cref{alg:WGAN}, we choose shallower networks with $L-1=2$ hidden layers and $W=256$ hidden units in each layer. We generate $T=50\,000$ samples for local enhancements and use a batch size $B=1\,280$ and $1000$ epochs for training. For the VB variant (B-GAN-VB in \Cref{alg:WGAN-VB}), we use the same network architecture and training configuration as B-GAN-RL. We set the weights of the VB loss to be $0.2$.

For the two ABC methods, we use the same $T=1\,000\,000$ pairs used in the B-GAN training. We adopt the summary statistics described in \Cref{sec:lv_main} for the naive ABC (SS) and Wasserstein version is calculated on top of the pairs of predator-prey population $\{x_t, y_t\}$. For both models, we again accept ABC draws with the top $1\%$  smallest data discrepancies.

For the SNL model, the author suggested building the network on top of the same set of summary statistics used in naive ABC, and 20 epochs of training. The network architecture and other training configuration remain the same as in \Cref{sec:implementation_toy}.

\subsection{Implementation Details for the Boom-and-Bust Example}\label{sec:implementation_bnb}

We generate $T=500\,000$ pairs of $(X_i, \theta_i)$ for training the vanilla B-GAN in \Cref{alg:WGAN}. We use the same batch size, learning rate and network architectures as in \Cref{sec:implementation_toy}, except that we train the networks for $2\,000$ epochs this time. We have explored three types of different inputs: (1) the time-series itself; (2) the summary statistics suggested by previous literature such as \citep{an2020robust} (described in \Cref{sec:bnb_main}); (3) the time-series together with the summary statistics. We find the one built on only the summary statistics works best. 

For this example, we find that this model is  more challenging than the Lotka-Volterra example and the vanilla posterior does not always learn the correct location of $\theta_0$. Directly using the B-GAN posterior results in poor training samples for both the 2-step refinement and the VB implementation. To improve the robustness of our methods in repeated experiments, we have revised the proposal distribution in the second step to be a mixture of 50\% prior and 50\% posterior from the vanilla B-GAN. This ensures that we can be guided towards the area closer to the true values $\theta_0$ while guarantees that $\theta_0$ is absolutely covered by the proposal distribution.

For the local enhancement variants, B-GAN-2S in \Cref{alg:WGAN-RL} and B-GAN-VB in \Cref{alg:WGAN-VB}, we generate $T=50\,000$ samples. The network architectures and training configurations are the same as the ones in \Cref{sec:implementation_toy}.

For ABC methods, they are built on the same $T=500\,000$ draws used in B-GAN. We adopt the summary statistics mentioned in \Cref{sec:bnb_main} for the naive ABC (SS), and the 2-Wasserstein ABC is trained on the time series.

For SNL models, we use the summary statistics as the input to the networks and train the model for 20 epochs, similar to the setting used in \Cref{sec:implementation_lv}.

\subsection{Implementation Details for the Common Cold example}\label{sec:implementation_cold}

We generate $T=500\,000$ pairs of $(X_i, \theta_i)$ for training the vanilla B-GAN in \Cref{alg:WGAN}. We use the same batch size, learning rate as in \Cref{sec:implementation_toy}. We  use ReLU networks with $L-1=2$ hidden layers with $(128,128)$ units from bottom to top for both the generator and the critic functions. 

For B-GAN-RL and B-GAN-VB, we generate $T=50\,000$ samples from the B-GAN posterior.The network architectures and training configurations are the same as B-GAN.

\subsection{Computation Costs Comparison}\label{sec:compute_cost}

We provide  comparisons of computation times of each method in \Cref{tab:comp_time} for the Gaussian example and the Lotka-Volterra model. Note that SS, W2 and SNL were executed with CPUs and all B-GAN models were computed using GPUs. We report the time of B-GAN-2S and B-GAN-VB for computation using the adjusted prior only (i.e. without the pilot run). Since B-GAN model learns the joint distribution of $(X, \theta)$ and is universal regardless of the observed data $X_0$, we recycle the same pre-trained B-GAN model to recover the adjusted prior repeatedly for different $X_0$ in our simulation study. Although we only report the computation time for one repetition here, this feature saves a lot of computation costs when one wants to investigate average performance from multiple repetitions.  

The complexity of computing the exact Wasserstein distance is $\mathcal O(T^3)$ \citep{burkard2009assignment} and  $\mathcal O(T^2)$ \citep{cuturi2013sinkhorn} for the  approximate one. The computation costs of SNL and B-GAN depends on the network architecture and  how fast the networks converge. We can only give a rough computation complexity estimate as $\mathcal O(T\times \# epochs \times \# weights)$ for these neural network based models.  From \Cref{tab:comp_time}, we observe that our methods could be more scalable than W2 and SS when the dimension of the dataset is high (8 of Gaussian model vs 402 of Lotka-Volterra (LV) model). The computation time of SNL on LV is smaller than for the Gaussian model due to fewer epochs in training (the author suggested 20 rounds for LV and 40 rounds for Gaussian model). In addition, SNL is trained on summary statistics ($q=9$) rather than the time-series for the LV model. The computation costs of SS increase significantly on the LV example, resulting from both the increase in dimension and the computation costs of the selected summary statistics. The computation costs of the Wasserstein distance are unsurprisingly high as it is known that they do not scale well.

\begin{table}[!ht]
\centering
\begin{tabular}{l | *{6}{r}}
\toprule
      & SS   & W2  & SNL           & B-GAN           & B-GAN-2S          & B-GAN-VB       \\
      \midrule
Gauss&33.75    & 221.28   & 4790.56 & 2736.93 & 676.25   & 726.22 \\
Lotka-Volterra    & 5846.95 & 162644.96 & 3080.96 & 1610.05           & 762.21 & 753.61\\
\bottomrule
\end{tabular}
\caption{\footnotesize Computation time of one repetition for each method on Gauss example and Lotka-Volterra (LV) example (in seconds). The time of B-GAN-2S and B-GAN-VB is for computation using the adjusted prior.}\label{tab:comp_time}
\end{table}

\end{document}